\documentclass[hidelinks, 12pt]{article}
\usepackage{amsmath}
\usepackage{amssymb}
\usepackage{amsthm}
\usepackage{relsize}
\usepackage[utf8]{inputenc}
\usepackage{mathtools}
\usepackage{url}
\usepackage[usenames,svgnames]{xcolor}
\usepackage{enumitem}
\usepackage[hang,flushmargin]{footmisc}

\usepackage[hang]{footmisc}

\usepackage[normalem]{ulem}
\makeatletter
\long\def\emph#1{\ifmmode\nfss@text{\em #1}\else\hmode@bgroup\text@command{#1}\em\check@icl #1\check@icr\expandafter\egroup\fi}
\makeatother
\usepackage{graphics,epsfig,cancel,amsthm,amsfonts,amssymb,amsmath,mathrsfs,upgreek,mathtools,subcaption,textcomp,ulem,mathdots,color,tikz,tikz-cd,todonotes,framed,epiolmec,phaistos,adjustbox,ulem,multicol,latexsym,ulem,multicol,latexsym,xparse,float,todonotes,etoolbox,upgreek}
\usepackage{nameref}
\usepackage{mdframed}
\usepackage{fancyvrb}
\usepackage{relsize}
\usepackage{mathtools}
\usepackage{url}

\newtheorem{theorem}{Theorem}[section]

\newtheorem{example}[theorem]{Example}

\newtheorem{lemma}[theorem]{Lemma}
\newtheorem{remark}{Remark}

\newtheorem{proposition}[theorem]{Proposition}

\newtheorem{definition}{Definition}

\newcommand{\tr}{\operatorname{tr}}

\usepackage{mathrsfs} 
\newcommand\Vtextvisiblespace[1][.5em]{%
  \mbox{\kern.06em\vrule height.3ex}%
  \vbox{\hrule width#1}%
  \hbox{\vrule height.3ex}}

\usepackage{cancel}

\usepackage{esint}
\usetikzlibrary{arrows,positioning,cd,calc,math,decorations.pathreplacing,decorations.markings,decorations.pathmorphing,patterns
}
\tikzset{wave/.style={decorate, decoration=snake}}

\usepackage{xcolor}
\definecolor{MyBlue}{rgb}{0.25,0.5,0.75}
\colorlet{NextBlue}{MyBlue!20}
\colorlet{SecondBlue}{MyBlue!40}
\usepackage{subcaption}
\NewDocumentCommand{\tens}{t_}
 {%
  \IfBooleanTF{#1}
   {\tensop}
   {\otimes}%
 }
\NewDocumentCommand{\tensop}{m}
 {%
  \mathbin{\mathop{\otimes}\displaylimits_{#1}}%
 }
 \definecolor{shadecolor}{rgb}{0.9, 0.9, 0.86}
\definecolor{darkgreen}{rgb}{0.2, 0.5,  0}
\definecolor{darkblue}{rgb}{0.1,0.1,0.45}
\definecolor{red}{rgb}{0.9,0,0}

\RequirePackage{doi}
\usepackage[all]{hypcap} 

\DeclareMathOperator{\res}{Res}
\newcommand{\sfa}{{\sf a}}

\newcommand{\sfd}{{\sf d}}

\newcommand{\sfY}{{\sf Y}}
\newcommand{\sfm}{{\sf m}}

\newcommand{\cC}{\mathcal{C}}
\newcommand{\HHh}{\hat{\mathcal{H}}}

\newcommand{\T}{\mathcal{T}}

\newcommand{\ad}{\mathrm{ad}}

\newcommand{\ph}{\hat{\psi}}
\newcommand{\phb}{\bar{\hat{\psi}}}

\usepackage{tikz}
\usepackage{upgreek}
\usetikzlibrary{decorations.markings}
\definecolor{darkspringgreen}{rgb}{0.05, 0.5, 0.06}
\definecolor{MyBlue}{rgb}{0.25,0.25,0.75}
\definecolor{MyRed}{rgb}{0.75,0.25,0.25}
\colorlet{NextBlue}{MyBlue!20}
\colorlet{SecondBlue}{MyBlue!40}

\def\XXint#1#2#3{{\setbox0=\hbox{$#1{#2#3}{\int}$ }
\vcenter{\hbox{$#2#3$ }}\kern-.6\wd0}}

\def \bea#1\eea {\begin{align} #1 \end{align}}
\def\red#1{\textcolor[rgb]{0.9, 0, 0}{#1} }

\def \bl{\begin{lemma}}
\def \el{\end{lemma}}

\def \br{\begin{remark}\rm }
\def \er{\hfill $\triangle$\end{remark}}

\def \bp{\begin{proposition}}
\def \ep{\end{proposition}}

\def \eqref#1{(\ref{#1})}

\def \remove #1 { \sout{ {\color{red} #1}}}

\newcommand{\C}{\mathbb{C}}

\def\res{\mathop{\mathrm {res}}\limits_}
\def \tr {\mathrm{tr}\,}
\def\be{\begin{equation}}
\def\ee{\end{equation}}
\def \bd{\begin{definition}}
\def \ed{\end{definition}}

\def\SL{\mathrm{SL}}
\def\GL{\mathrm{GL}}

\DeclareMathOperator{\Span}{span}

\DeclareMathOperator{\SO}{SO}
\DeclareMathOperator{\Sp}{Sp}

\def\ra{{\rightarrow}}

\def\tr{\mathrm {tr}}

\def\det{\mathrm {det}}

\def\Pf{\mathrm {Pf}}

\def\res{\mathop{\mathrm {res}}\limits}

\def\JMU{\mathrm {JMU}}

\def\res{\mathop{\mathrm{res}}\limits}

\def\&{&{\hskip -20 pt}}

\def\be{\begin{equation}}
\def\ee{\end{equation}}

\def\bea{\begin{eqnarray}}
\def\eea{\end{eqnarray}}

\def\bt{\begin{theorem}}
\def\et{\end{theorem}}

\def\bex{\begin{example}\small \rm}
\def\eex{\end{example}}

\def\bexs{\begin{examples}\small \rm}
\def\eexs{\end{examples}}

\def\ra{\rightarrow}
\def\ss{\subset}

\def\br{\begin{remark}\small \rm}
\def\er{\end{remark}}

\def\FF {{\mathcal F}}

\def\HH{{\mathcal H}}

\def\SS{{\mathcal S}}

\def\Ib{\mathbf{I}}

\def\nb{\mathbf{n}}
\def\pb{\mathbf{p}}
\def\qb{\mathbf{q}}

\def\tb{\mathbf{t}}

\def\Cbb{\mathbb{C}}

\def\Zbb{\mathbb{Z}}

 \def\grg{\mathfrak{g}}
 \def\grh{\mathfrak{h}}

 \def\grso{\mathfrak{so}}

\newmuskip\pFqmuskip

\newcommand*\pFq[6][8]{%
  \begingroup 
  \pFqmuskip=#1mu\relax
  \mathchardef\normalcomma=\mathcode`,
  \mathcode`\,=\string"8000
  \begingroup\lccode`\~=`\,
  \lowercase{\endgroup\let~}\pFqcomma
  {}_{#2}F_{#3}{\left[\genfrac..{0pt}{}{#4}{#5};#6\right]}%
  \endgroup
}
\newcommand{\pFqcomma}{{\normalcomma}\mskip\pFqmuskip}

\numberwithin{equation}{section}
\numberwithin{remark}{section}

\begin{document}
\baselineskip 16pt
 
 \medskip
\begin{center}
\begin{Large}\fontfamily{cmss}
\fontsize{17pt}{27pt}
\selectfont
	\textbf{Fredholm Pfaffian $\tau$-functions for orthogonal isospectral and isomonodromic systems}
	\end{Large}

\renewcommand{\thefootnote}{\fnsymbol{footnote}}

\bigskip \bigskip
\begin{large}
M. Bertola$^{1,2,3}$\footnote[2]{e-mail:marco.bertola@concordia.ca},
Fabrizio  Del Monte$^{1,2}$\footnote[3]{e-mail:delmonte@crm.umontreal.ca,}
and J. Harnad$^{1,2}$\footnote[1]{e-mail:harnad@crm.umontreal.ca}
\end{large}
\\
\bigskip
\begin{small}
$^{1}${\em Centre de recherches math\'ematiques, Universit\'e de Montr\'eal, \\C.~P.~6128, succ. centre ville, Montr\'eal, QC H3C 3J7  Canada}\\
$^{2}${\em Department of Mathematics and Statistics, Concordia University\\ 1455 de Maisonneuve Blvd.~W.~Montreal, QC H3G 1M8  Canada}\\
$^{3}${\em SISSA, International School for Advanced Studies, via Bonomea 265, Trieste, Italy} 
\end{small}
\end{center}

\begin{abstract}
We extend the approach to  $\tau$-functions as Widom constants developed by Cafasso, Gavryleko and Lisovyy
to orthogonal loop group Drinfeld-Sokolov reductions and isomonodromic deformations systems. The
combinatorial expansion of the $\tau$-function as a sum of correlators, each expressed as
products of finite determinants, follows from using multicomponent fermionic vacuum expectation values of certain 
dressing operators  encoding the initial conditions and  dependence on the time parameters.
When reduced to the orthogonal case, these correlators become finite Pfaffians and the 
determinantal $\tau$-functions, both in the Drinfeld-Sokolov and isomonodromic case,  become squares of $\tau$-functions of
Pfaffian type. The results are illustrated by several examples, consisting of polynomial $\tau$-functions of
orthogonal Drinfeld-Sokolov type and isomonodromic ones with four regular singular points.\end{abstract}
\bigskip
\renewcommand{\thefootnote}{\arabic{footnote}}
\tableofcontents
\section{Introduction}
\label{intro}

 In \cite{Cafasso2017,  Cafasso2018} a new approach to  $\tau$-functions for integrable systems,  viewed as  {\em Widom constants}
associated to various classes of Riemann-Hilbert problems, was applied both to Drinfeld-Sokolov hierarchies 
and isomonodromic deformation systems. This led to an explicit combinatorial expansion of the $\tau$-functions
as sums over products of pairs of determinants characterizing the initial condition data and the deformation parameter dependence,
respectively. 

In this work, we develop this approach further, interpreting the $\tau$-function, as in the Sato-Jimbo-Miwa  approach \cite{SJM77, SJM78, Sato81, MJD},
 in terms of fermionic free field vacuum expectation values (VEV's), and then applying a suitable reduction procedure to obtain
integrable systems  in orthogonal loop groups and algebras. The determinant  expansions are replaced by infinite sums over products of pairs of finite Pfaffians, 
  with all the terms interpreted equivalently as fermionic VEV's, or scalar products in the fermionic Fock space of suitably ``dressed'' elements.
  
  Section \ref{ortho_reductions} defines a family of (complex) scalar products $Q_S$ on our Hilbert space $\HH^N$, 
    parametrized by symmetric $N\times N$ matrices $S$ whose square is the unit matrix.
  These are used to  define orthogonal reduction from the general loop group setting $L\GL(N)$ to the orthogonal loop groups $L_SSO(N)$,

In Section \ref{widom_constant_tau}, the Riemann-Hilbert problem definition of the {\em Widom constant} $\tau$-function $\tau_W[J]$
associated to a loop group element $J \in L\GL(N)$ is recalled \cite{Cafasso2018}, and expressed as the Fredholm determinant 
of a trace zero perturbation of the identity operator. The upper and lower block triangular parts, when expressed in a Fourier basis,
are labelled by $N$-tuples of  {\em Maya diagrams}  (see \cite{MJD} or \cite{HB},  App. A)  or, equivalently, {\em charged Young diagrams}, 
representing pairs $(\lambda, n)$ of a partition $\lambda$ and an integer $n$, denoting the  fermionic charge , 
which is the ``center of gravity'' of the distribution of occupied and unoccupied sites (``particles'' and ``holes'') along
 the $1D $ half-integer lattice.   
The restriction (\ref{eq:ORHP}) to  the orthogonal loop subgroup $L_SS O(N)$ is introduced, and the Fourier components of the 
upper and lower triangular blocks are similarly labelled by $N$-tuples of pairs, but consisting now
of {\em strict} partitions.  The resulting skew symmetric conditions implied on 
the upper and lower triangular parts of the integral operators  defining the trace zero perturbations are derived (Lemma \ref{eq:antisymB}).

The associated Clifford algebra of multicomponent fermions acting on the  fermionic Fock space is introduced in 
Section \ref{sec:N_comp_fermion_LGLN_tau}. The upper and lower triangular blocks of the integral operators  
are used to define ``dressed states'' (\ref{d_state}), (\ref{a_state}),  whose scalar product reproduces the $\tau$-function 
(Proposition \ref{ad_scalar_prod}). The combinatorial expansion (\ref{eq:WidomMinor}) of the Widom $\tau$-function  $\tau_W[J]$ as a sum over products of 
finite dimensional determinants derived in \cite{Cafasso2018} is rederived using this fermionic representation.

The two main new results of the paper are contained in Theorems \ref{thm:RealFTau} and \ref{thm:PfMinor}.
Theorem \ref{thm:RealFTau} expresses the  Pfaffian-type $\tau$-function  $\tau_O[J]$, 
as the scalar product (\ref{eq:TauD}) of two  ``dressed''  fermionic states, analogously to (\ref{d_state}), (\ref{a_state}),
and shows that the corresponding determinantal Widom  $\tau$-function $\tau_W[J]$ is its square (\ref{tau_widom_tau_pfaff}). 
Theorem \ref{thm:PfMinor}  is the Pfaffian analogue of  the combinatorial expansion (\ref{eq:WidomMinor}), replacing it with the 
sum (\ref{eq:TauFredPf})  over products of finite Pfaffians.  The sense in which $\tau_O[J]$ is an infinite dimensional Pfaffian is 
explained in the subsequent discussion,  leading to formula (\ref{eq:TauFredPf}).
 
 In Section \ref{drinfeld_sokolov_pfaff}, we recall the Drinfeld-Sokolov hierarchiy for simple  loop Lie algebras, and in Section \ref{drinfeld_sokolov_tau}, 
 compute some explicit examples of polynomial $\tau$-functions for the orthogonal loop algebras  corresponding
 to the affine Kac-Moody algebras $B_1^{(1)}$, $B_2^{(1)}$ and $D^{(1)}_4$.

Section \ref{SO_N_isomonodromy} concerns the isomonodromic case with four simple poles. It is shown,
by a change of representation, how the results of \cite{Cafasso2017} for the case $\SL(2)$ can
 be adapted to deducing the corresponding $\SO(3)$ isomonodromic $\tau$-function
 as a Pfaffian which, in turn, is the square of the $\SL(2)$ $\tau$-function computed in \cite{Cafasso2017} in terms
 of sums over hypergeometric functions.

\section{Definition of the spaces and orthogonal reductions}
\label{ortho_reductions}

Consider the Hilbert space 
\be
\HH^N:= L^2(S^1,\Cbb^N),
\ee
 with basis elements 
\be
  e_r^\alpha:= z^{r-1/2}\hat{e}_\alpha,   \quad \alpha \in \{1, \dots, N\}, \quad r \in \Zbb' := \Zbb+\frac{1}{2},
\ee
 where
\be
  z:=e^{i\theta} \in S^1,
\ee
and $\hat{e}_\alpha$ is the unit vector in $\mathbb{C}^N$ with $\alpha$-th entry $1$ and the others $0$. The dual basis elements of $\HH^{N*}$ 
are denoted $\{e^r_\alpha\}$, defined such that
\be
e^r_\alpha(e_s^\beta )=\delta_\alpha^\beta\delta_s^r.
\ee
This pairing may be viewed equivalently as defining a scalar product $Q$ on $\HH^N\oplus\HH^{N*}$  by:
\bea
Q&\&:(\HH^N\oplus\HH^{N*}) \times (\HH^N\oplus\HH^{N*})  \rightarrow  \Cbb  \cr
 Q\left((F+\mu),(G+\nu) \right)&\&:=\mu(G)+\nu(F), \quad  F,G \in \HH^N, \ \mu, \nu \in \HH^{N*},
\eea
with respect to which both $\HH^N$ and $\HH^{N*}$ are totally isotropic.
It also establishes an isomorphism $\HH^N\simeq \HH^{N*}$ in which we identify the basis elements 
with their duals
\be\label{eq:DualPair}
e_r^\alpha \sim e^{-r}_\alpha,
\ee
so that
\be
e^r_\alpha \simeq z^{-r-\frac{1}{2}}\hat{e}_\alpha.
\ee

For any symmetric $N\times N$ matrix involution $S$,   define two mutually orthogonal complementary 
subspaces $\HH^N_{(S)},\,\HHh^N_{(S)}$ as
\be
\HH^N_{(S)}:= \Span\left\{f_{r,\alpha}^{(S)} \right\}, \quad
\HHh^N_{(S)}:= \Span\left\{\hat{f}_{r,\alpha}^{(S)} \right\},
\ee
where
\be
f_{r,\alpha}^{(S)} := \frac{1}{\sqrt{2}}\left(e^{-r}_\alpha+ \sum_{\beta=1}^N S_{\alpha \beta} e^\beta_{r} \right), \quad \hat{f}_{r,\alpha}^{(S)}
:= \frac{i}{\sqrt{2}}\left(e^{-r}_\alpha-\sum_{\beta=1}^N S_{\alpha \beta} e^\beta_{r} \right).
\ee
This gives an alternative  direct sum decomposition
\be
\HH^N\oplus\HH^{N*}=\HH^N_{(S)}\oplus\HHh^N_{(S)}.
\ee
Restricting  $Q$ to $\HH^N_{(S)}$ and $\HHh^N_{(S)}$ induces a scalar product on each
\bea
Q_S &\&:= Q|_{\mathcal{H}^N_{(S)}},  \quad \hat{Q}_S:= Q|_{\HHh^N_{(S)}},\\
&\& \cr
Q_S\left(f_{r,\alpha}^{(S)},f_{s,\beta}^{(S)} \right)&\&=S_{\alpha\beta}\delta_{r,-s}, 
\quad \hat{Q}_{S}\left(\hat{f}_{r,\alpha}^{(S)},\hat{f}_{s,\beta}^{(S)} \right) = S_{\alpha\beta}\delta_{r,-s}.
\eea
The spaces $\HH^N_{(S)}$ and $\HHh^N_{(S)}$ may both be viewed as isomorphic to $\mathcal{H}^N=L^2(S^1,\Cbb^N)$, through the correspondence of basis elements
\be
f^{(S)}_{r,\alpha}\leftrightarrow e_r^\alpha, \quad \hat{f}^{(S)}_{r,\alpha}\leftrightarrow e_r^\alpha.
 \label{eq:QO_isom}
\ee

In the following, the symbol $\oint$  denotes the positively oriented contour integral
$\oint_{S^1}$ around the unit circle $S^1 = \{e^{i\theta} \ss \Cbb\}$ centred at the origin. 
Under the isomorphism $\HH^N \sim \HH^N_S$,  the bilinear form  $Q_{S}:\HH^N\times\HH^N\ra \Cbb$ is given by
\be
 Q_{S}(F,G) = \frac{1}{2\pi i} \oint F(z)^t\,S G(z) dz=: \left(F,G \right)_{S} ,
 \label{eq:PairingO}
 \ee
where the superscript $^t$ denotes matrix transposition 
The loop group elements that preserve this scalar product consist of invertible matrix-valued functions on $S^1$ satisfying
\begin{equation}
\label{eq:OSJump}
J(z)^tSJ(z)=S.
\end{equation}
We refer to the loop group arising from this reduction as the ``$S$-orthogonal loop group'', and denote it by $L_SSO(N)$.

Let $\HH^N_+, \HH^N_- \ss \HH^N$ denote the subspaces consisting of elements admitting analytic continuation inside and outside the circle $S^1$, respectively, with the elements of $\HH^N_-$ vanishing at $z=\infty$. 
We then have the direct sum decomposition
\begin{equation}
\HH^N=\HH^N_+\oplus\HH^N_-,
\end{equation}
where each subspace is totally isotropic with respect to $Q_{S}$, and (\ref{eq:PairingO})  defines a dual pairing between the two. 
 Expressed as Fourier series, we have
\be
F(z)=\sum_{r\in\Zbb'_+} F_pz^{p-\frac{1}{2}}\in\HH^N_+, \quad G(z)=\sum_{p\in \Zbb'_+} G_pz^{-p-\frac{1}{2}}\in \HH^N_-,
\ee
where $\Zbb'_+$ is the set of positive half-integers.


\section{The Widom $\tau$-function and operator transposes}
\label{widom_constant_tau}
Following \cite{Cafasso2017}, we  consider matrix Riemann-Hilbert problems (RHPs) on the circle $S^1$: Let
\be
\Psi_+(z)\in L_+\GL(N,\mathbb{C}) \ss L\GL(N,\mathbb{C}) 
\ee
 denote the elements of the subgroup of the loop group $L\GL(N,\mathbb{C}) $ that are  analytic inside the circle, and 
 \be
\Psi_-(z)\in L_-\GL(N,\mathbb{C}) \ss L\GL(N,\mathbb{C}) 
\ee 
those that are analytic outside, with $\Psi_-(\infty) = \Ib$. The jump across the  circle  is denoted
\begin{equation}
\label{eq:JumpRHP}
J(z):= \Psi_-(z)^{-1}\Psi_+(z).
\end{equation}
A factorization of this type will be called a {\em direct} factorization, as opposed to a factorization of the type
\begin{equation}
\label{eq:JumpDual}
J(z)=\bar{\Psi}_+(z)\bar{\Psi}_-(z)^{-1},
\end{equation}
which will be called the {\em dual} factorization, and the corresponding RHP the dual RHP.

The action of an integral operator $A:\HH^N\ra \HH^N$ with matrix kernel $A(z,w)$ is expressed as
\be
A(F)(z)=\frac{1}{2\pi i}\oint A(z,w)F(w) dw,  \quad F\in \HH^N,
\ee
and  the matrix transpose of $A(z,w)$ is denoted $A(z,w)^t$. 
Recall the definition of the Widom $\tau$-function  \cite{Cafasso2017}:
\begin{definition}
Given a RHP on the  unit circle $S^1= \{z = e^{i\theta}\}$ with jump \eqref{eq:JumpRHP}, the Widom $\tau$-function is
\begin{equation}\label{eq:WidomTauDef}
\tau_W[J]:=\det_{\HH^N_+}\left(\Pi_+J^{-1}\Pi_+J \right)=\det_{\HH^N}\left[\mathbb{I}+\left( \begin{array}{cc}
0 & \sfa \\
\sfd & 0
\end{array} \right) \right]=\det_{\HH^N_+}\left(\mathbb{I}-\sfa\sfd \right) ,
\end{equation}
where $\mathbb{I}$ is the identity operator on $\HH^N$, $\Pi_+$ is the projection operator of $\HH^N$ onto $\HH^N_+$ along $\HH^N_-$ and
 \be
 \sfa:\HH^N_-\rightarrow \HH^N_+ , \quad \sfd:\HH^N_+\rightarrow \HH^N_-
 \ee
  have kernels
\be
\label{eq:adKers}
\sfa(z,w):=\frac{\mathbb{I}_N-\Psi_+(z)\Psi_+(w)^{-1}}{z-w}, \quad \sfd(z,w):=\frac{\Psi_-(z)\Psi_-(w)^{-1}-\mathbb{I}}{z-w}.
\ee
Their Fourier expansions are written as:
\be
\label{eq:adFourier}
\sfa(z,w)=\sum_{p,q\in\Zbb'_+}\sfa_{-q}^pz^{p-\frac{1}{2}}w^{q-\frac{1}{2}}, \quad \sfd(z,w)=\sum_{p,q\in\Zbb'_+}\sfd^{-q}_pz^{-q-\frac{1}{2}}w^{-p-\frac{1}{2}}.
\ee
\end{definition}

We use the same notation  
\be
\sfa: \HH^N  \rightarrow \HH^N, \quad \sfd: \HH^N  \rightarrow \HH^N
\ee
 to denote the extension of the operators $\sfa$ and $\sfd$  to $\HH^N$ defined by setting the restriction of $\sfa$ to $\HH^N_+$ and  $\sfd$ to $\HH^N_-$ to vanish. 
The condition \eqref{eq:OSJump} that $J(z)$  preserve the scalar product $Q_S$  carries over to $\Psi_\pm$:
\be
J(z)SJ(z)^t=\Psi_\pm(z)S\Psi_\pm(z)^t=S.
\label{eq:ORHP}
\ee
We then have the following:
\begin{lemma}
If the reduction condition \eqref{eq:OSJump} corresponding to the scalar product $Q_S$ is satisfied,
 the integral kernels of the operators $\sfa,\sfd$ in \eqref{eq:adKers} satisfy
\be
\label{eq:antisymB}
\sfa(w, z)^t=-S\sfa(z,w)S, \quad \sfd(w,z)^t=-S\sfd(z,w)S.
\ee
In terms of the coefficients appearing in the Fourier expansion, this is equivalent to 
\be
\label{eq:antisymB2}
(\sfa^{p}_{-q})_{\alpha\beta}=-(S\sfa^{q}_{-p}S)_{\beta\alpha},\quad (\sfd^{-q}_{p})_{\alpha\beta}=-(S\sfd^{-p}_{q}S)_{\beta\alpha}.
\ee
\end{lemma}
\begin{proof}
By direct computation:
\begin{equation}
\begin{split}
\sfa(w,z)^t & =\frac{\mathbb{I}_N-\Psi_+(z)^{-t}\Psi_+(w)^t}{w-z} \\
& \mathop{=}^{\eqref{eq:ORHP}}-\frac{\mathbb{I}_N-S\Psi_+(z)\Psi_+(w)^{-1}S}{z-w} \\
& =-S\sfa(z,w)S,
\label{eq:alemma}
\end{split}
\end{equation}
\begin{equation}
\begin{split}
\sfd(w,z)^t & =\frac{\mathbb{I}_N-\Psi_-(z)^{-t}\Psi_-(w)^t}{w-z} \\
& \mathop{=}^{\eqref{eq:ORHP}}-\frac{\mathbb{I}_N-S\Psi_+(z)\Psi_+(w)^{-1}S}{z-w} \\
&=-S\sfd(z,w)S.
\end{split}
\end{equation}
Eq.~\eqref{eq:antisymB2} is just the expression of the above equations  in terms of Fourier coefficients. For example, the LHS of \eqref{eq:alemma} reads
\be
\sfa(w,z)_{\beta\alpha}\mathop{=}^{\eqref{eq:adFourier}}\sum_{p,q\in\mathbb{Z}_+'}(\sfa^{p}_{-q})_{\beta\alpha}w^{p-\frac{1}{2}}z^{q-\frac{1}{2}},
\ee
while the RHS is
\bea
-S\sfa(z,w)S\  &\& =-\sum_{p,q\in\mathbb{Z}_+'}(S\sfa^p_{-q}S)_{\alpha\beta}z^{p-\frac{1}{2}}w^{q-\frac{1}{2}}\\
&\& =-\sum_{p,q\in\mathbb{Z}_+'}(S\sfa^q_{-p}S)_{\alpha\beta}z^{q-\frac{1}{2}}w^{p-\frac{1}{2}}.
\eea
Equating the two we obtain
\begin{equation}
(\sfa^{p}_{-q})_{\alpha\beta}=-(S \sfa^{q}_{-p}S)_{\beta\alpha}.
\end{equation}
The second equation of \eqref{eq:antisymB2} follows similarly.
\end{proof}


\section{Fermionic VEV representation and  orthogonal \hbox{reductions}}
\label{fermionic_widom_tau}


\subsection{$N$-component fermions and $L \GL(N)$ $\tau$-functions}
\label{sec:N_comp_fermion_LGLN_tau}

We recall  the  $N$-component fermionic Fock space $\FF^N$ of charged free fermions \cite{jimbo_miwa_83}. 
The vacuum vector in the fermionic charge $0$ sector $\FF^N_0$
is denoted $|0 \rangle$ and the dual vacuum vector $\langle 0 |$. The fermionic creation and annihilation
operators,  $\{\psi^{(\alpha)}_m, \psi^{\dag (\alpha)}_m\}_{m\in \Zbb,\  (\alpha) \in \{1, \dots , N\}}$,
satisfy the anticommutation relations
\be
[\psi^{(\alpha)}_m, \psi^{\dag (\beta)}_n ]_+ = \delta_{\alpha, \beta} \delta_{m, n},
\quad [\psi^{(\alpha)}_m, \psi^{\beta}_n ]_+  = [\psi^{\dag \alpha}_m, \psi^{\dag \beta}_n ]_+ =0,
\ee
and the vacuum annihilation conditions
\bea
\psi^{(\alpha)}_{-m-1} |0 \rangle &\&=0, \quad \psi^{\dag (\alpha)}_m |0 \rangle =0,  \cr
&\& \cr
\langle 0 |\psi^{(\alpha)}_{m}&\&=0, \quad \langle 0 |\psi^{\dag (\alpha)}_{-m-1}=0,  \quad m \ge 0.
\eea

In the following, we use a different notational convention, in which the creation and
annihilation operators are labelled by $\frac{1}{2}$-integers \hbox{$r \in \Zbb'$} and
 denoted
$\{\hat{\psi}^\alpha_r, \bar{\hat{\psi}}^\alpha_r\}_{r\in \Zbb', \alpha \in \{1, \dots, N\}}$.
These are related to the ones above as  follows:
\be
\hat{\psi}_r^\alpha := \psi_{-r-1/2}^{(\alpha)}, \quad \bar{\hat{\psi}}_r^\alpha := \psi^{\dag (\alpha )}_{r-1/2},
\quad r \in \Zbb', \quad \alpha \in \{1, \dots , N\}.
\ee
They satisfy the anticommutation relations
\be
[\hat{\psi}_r^\alpha,\bar{\hat{\psi}}_s^\beta ]_+=\delta^{\alpha\beta}\delta_{r,-s}, \quad
[\hat{\psi}_r^\alpha,\hat{\psi}_s^\beta ]_+=0 ,\quad 
[\bar{\hat{\psi}}_r^\alpha,\bar{\hat{\psi}}_s^\beta ]_+=0,
\label{eq:cfanticom}
\ee
and vacuum annihilation conditions:
\bea
\label{eq:ComplexVacuum}
\hat{\psi}_r^\alpha|0\rangle &\& = 0, \quad \bar{\hat{\psi}}_r^\alpha|0\rangle = 0 , \cr
&\& \cr
\langle 0 | \hat{\psi}_{-r}^\alpha &\&= 0, \quad \langle 0| \bar{\hat{\psi}}_{-r}^\alpha = 0 ,\quad r\in\Zbb_+',
\eea
where we denote the set of positive half-integers by $\Zbb'_+$. We also define the  fermionic free fields
as generating functions for these:
\bea
\label{eq:cffField}
\hat{\psi}^\alpha(z)&\&:=\sum_{p\in\mathbb{Z'}}\hat{\psi}^{\alpha}_{p}z^{-p-\frac{1}{2}}, \cr
 \bar{\hat{\psi}}^\alpha(z) &\& :=\sum_{p\in\mathbb{Z'}}\bar{\hat{\psi}}^{\alpha}_{p}z^{-p-\frac{1}{2}}.
\eea

We sometimes will use the notation $\hat{\psi}$ and $ \bar{\hat{\psi}}$ to denote the doubly indexed components
$\{\hat{\psi}^\alpha_p\}_{\alpha \in \{1, \dots, N\}, \   p \in \Zbb'}$ and $\{\bar{\hat{\psi}}^\alpha_p\}_{\alpha \in \{1, \dots, N\}, \   p \in \Zbb'}$,
respectively, and $\sfa$ and $\sfd$ to denote the doubly indexed matrices $\left\{(\sfa^p_{-q})_{\alpha \beta}\right\}$ and $\left\{(\sfd^{-q}_p)_{\alpha \beta}\right\}$,
where
\be
(\sfa^p_{-q})_{\alpha \beta}  = 0, \quad (\sfd^{-q}_{p})_{\alpha \beta} =0 \quad \text{if} \ q \in  -\Zbb'_+ \  \text{or} \  p \in -\Zbb'_+ .
\ee
Denote by $\sfa^p_{-q}$ and  $\sfd^{-q}_p$ the $N\times N$ matrices with elements $(\sfa^p_{-q})_{\alpha \beta}$ and  $(\sfd^{-q}_p)_{\alpha \beta}$.  

The Clifford module (Fock space) may be viewed as the irreducible module of the Clifford algebra obtained by applying negative 
mode operators to the vacuum state $|0\rangle$. An orthonormal basis $\{ |\{p_{\alpha,i},q_{\beta,j}\}\rangle\}$ for the Fock space is defined by
\begin{equation}
|\{p_{\alpha,i},q_{\beta,j}\}\rangle:=\prod_{\alpha=1}^N\left(\prod_{i=1}^{\#\{p_{\alpha,i}\}}\hat{\psi}_{-p_{\alpha,i}}^\alpha\prod_{j=1}^{\#\{q_{\alpha,j}\}}\bar{\hat{\psi}}_{-q_{\alpha,j}}^\alpha \right)|0\rangle,
\label{ortho_basis}
\end{equation}
where $\#\{p_{\alpha,i}\}$, $\#\{q_{\alpha,i}\}$ denote the cardinalities of these sets for each given $\alpha$, and
the order of the factors in the product is defined to be such that $\alpha, i$ and $j$ all increase to the right. 
We can interpret the indices $\{p_{\alpha,i},q_{\alpha,j}\}$ as the positions of particles and holes in an N-tuple 
of Maya diagrams $\vec{\sfm}:=\left(m_1,\dots,m_N \right)$ (see \cite{MJD} or \cite{HB},  App. A). Equivalently,
 we may represent these by {\em charged} Young diagrams  $\vec{\sfY}_{\nb}= \left((\sfY_1, n_1), \dots, (\sfY_N, n_N)\right)$, 
 as illustrated in Figure \ref{fig:charged_young_diag}.
 The  fermionic charge $n_\alpha$ associated to the Young diagram $\sfY_\alpha$ is the difference between the  number of positions occupied
by particles (occupied positive $1/2$ integers)  and holes (unoccupied negative $1/2$ integers) of the Maya diagram, 
\begin{equation}
n_\alpha:=\#\{p_{\alpha,i}\}- \#\{q_{\alpha,i}\},\quad \alpha =1, \dots, N,
\end{equation}
and is indicated in Figure \ref{fig:charged_young_diag} by the position of the red diagonal line.

The rule for determining the $p_{\alpha,i}$'s and $q_{\alpha,i}$'s is the following.
 Place the Young diagram in the 4th quadrant of the Cartesian plane, adjacent to the axes, 
with squares of unit size, and take its union with the 1st and 3rd quadrants.
Within this region,  draw the diagonal line segment of slope  $= -1$ through the lattice point $(-n_\alpha, 0)$.  
This then adds to the region of the Young diagram a right angle triangle with leg lengths  $= | n_\alpha |$, 
either in the first quadrant $(n_\alpha <0)$ or the third quadrant $(n_\alpha>0)$,
touching both axes, defining an extended  polygon.
The $p_{\alpha, i}$'s are then the horizontal areas, within the extended polygon, of each lattice row, 
to the right of the diagonal segment and the $q_{\alpha,i}$'s are the vertical areas  of each lattice column,
below the diagonal.

In the zero fermionic charge case,   $n_\alpha=0$,  $\{p_{\alpha,i}, q_{\alpha,i}\}_{1\le \alpha \le N \atop 1\le i \le r_\alpha}$ 
coincide with $ 1/2 \ +$ the Frobenius indices  $(a_1, \dots, a_{r_\alpha} | b_1, \dots, b_{r_\alpha})$
of the partition $\lambda^{(\alpha)}$ corresponding to the Young diagram $\sfY_\alpha$, where $r_\alpha$ is the number of elements
along the main diagonal (the Frobenius rank); i.e.,  the number of squares to the right and below the  main diagonal,  
respectively, as in Figure \ref{fig:ChargelessFrob}. 
In the charged case, the diagonal has to be shifted by the charge: one counts the number of squares to the right and below the shifted diagonal, including squares that lie outside of the Young diagram, as in Figure \ref{fig:ChargedFrob}.
\begin{figure}
\begin{center}
\begin{subfigure}{0.45\textwidth}
    \centering
\begin{tikzpicture}[scale=2.5]
\draw[thin,darkspringgreen,pattern color=darkspringgreen,pattern=north east lines] (0.,0.) -- (2,0.) -- (2,-0.4) -- (0.4,-0.4)--(0,0) (2.4,-0.2) node{$p_1=\frac{9}{2}$};
\draw[thin,blue,pattern color=blue,pattern=north east lines] (0.4,-0.4) -- (1.6,-0.4) -- (1.6,-0.8) -- (0.8,-0.8)  (2.,-0.6) node{$p_2=\frac{5}{2}$};
\draw[thin,red,pattern color=red,pattern=north east lines] (0.8,-0.8) -- (1.2,-0.8) -- (1.2,-1.2)  (1.6,-1.) node{$p_3=\frac{1}{2}$};
\draw[thin,violet,pattern color=violet,pattern=north east lines] (0.,0.) -- (0,-1.6) -- (0.4,-1.6) -- (0.4,-0.4)--(0,0) (0.1,-1.8) node{$q_1=\frac{7}{2}$};
\draw[thin,orange,pattern color=orange,pattern=north east lines] (0.4,-0.4) -- (0.4,-1.6) -- (0.8,-1.6) -- (0.8,-0.8)  (0.6,-1.8) node{$q_2=\frac{5}{2}$};
\draw[thin,purple,pattern color=purple,pattern=north east lines] (0.8,-0.8) -- (0.8,-1.6) -- (1.2,-1.6) -- (1.2,-1.2)  (1.1,-1.8) node{$q_3=\frac{3}{2}$};
\draw[clip](0,0) -- ++(2,0) -- ++(0,-0.4)-- ++(-0.4,0)-- ++(0,-0.4)-- ++(-0.4,0)-- ++(0,-0.8)--(0,-1.6) --cycle;
\draw[clip](0,0) -- ++(2,0) -- ++(0,-0.4)-- ++(-0.4,0)-- ++(0,-0.4)-- ++(-0.4,0)-- ++(0,-0.8)--(0,-1.6) --cycle;
\draw[scale=0.4](0,-1.6*3) grid (5,5);
\draw[ultra thick,red] (0,0)--(1.2,-1.2) ;
\end{tikzpicture}
\caption{Young diagram corresponding to the state 
 $\ph_{-\frac{9}{2}}\ph_{-\frac{5}{2}}\ph_{-\frac{1}{2}} \phb_{-\frac{7}{2}}\phb_{-\frac{5}{2}}\phb_{-\frac{3}{2}}|0\rangle $
 of  fermionic charge $n=0$.  The $p_i$'s and $q_i$'s are the areas of the
 horizontal and vertical strips of given colour.\\ \\ \\ \\}
 \label{fig:ChargelessFrob}
\end{subfigure}
\hfill
\begin{subfigure}{0.45\textwidth}
    \centering
\begin{tikzpicture}[scale=2.5]
\draw[thin,darkspringgreen,pattern color=darkspringgreen,pattern=north east lines] (-0.4,0.) -- (2,0.) -- (2,-0.4) -- (0,-0.4)--(-0.4,0) (2.4,-0.2) node{$p_1=\frac{11}{2}$};
\draw[thin,blue,pattern color=blue,pattern=north east lines] (0,-0.4) -- (1.6,-0.4) -- (1.6,-0.8) -- (0.4,-0.8)  (2.,-0.6) node{$p_2=\frac{7}{2}$};
\draw[thin,red,pattern color=red, pattern=north east lines] (0.4,-0.8) -- (1.2,-0.8) -- (1.2,-1.2) -- (0.8,-1.2)  (1.6,-1.) node{$p_3=\frac{3}{2}$};
\draw[thin,gray, pattern color=gray, pattern=north east lines] (0.8,-1.2) -- (1.2,-1.2) -- (1.2,-1.6) -- (0.8,-1.2)  (1.6,-1.4) node{$p_4=\frac{1}{2}$};
q coordinates/particles
\draw[thin,violet,pattern color=violet,pattern=north east lines] (0.,-0.4) -- (0,-1.6) -- (0.4,-1.6) -- (0.4,-0.8) (0.1,-1.8) node{$q_1=\frac{5}{2}$};
\draw[thin,orange,pattern color=orange,pattern=north east lines] (0.4,-0.8) -- (0.4,-1.6) -- (0.8,-1.6) -- (0.8,-1.2)  (0.6,-1.8) node{$q_2=\frac{3}{2}$};
\draw[thin,purple,pattern color=purple,pattern=north east lines] (0.8,-1.2) -- (0.8,-1.6) -- (1.2,-1.6)   (1.1,-1.8) node{$q_3=\frac{1}{2}$};
\draw[ultra thick,red] (-0.4,0)--(1.2,-1.6) ;
\draw[clip](0,0) -- ++(2,0) -- ++(0,-0.4)-- ++(-0.4,0)-- ++(0,-0.4)-- ++(-0.4,0)-- ++(0,-0.8)--(0,-1.6) --cycle;
\draw[clip](0,0) -- ++(2,0) -- ++(0,-0.4)-- ++(-0.4,0)-- ++(0,-0.4)-- ++(-0.4,0)-- ++(0,-0.8)--(0,-1.6) --cycle;
\draw[scale=0.4](0,-1.6*3) grid (5,5);
\end{tikzpicture}
\caption{The same Young diagram extended to correspond to the state
$ \ph_{-\frac{11}{2}}\ph_{-\frac{7}{2}}\ph_{-\frac{3}{2}}\ph_{-\frac{1}{2}} \phb_{-\frac{5}{2}}\phb_{-\frac{3}{2}}\phb_{-\frac{1}{2}}|0\rangle$ 
 of fermionic charge $n=1$. This is encoded in a shift of the diagonal with respect to which the $p_i$'s and $q_i$'s 
 are computed. \\ \\ \\  }
 \label{fig:ChargedFrob}
\end{subfigure}

\centering
\begin{subfigure}{0.6\textwidth}
\begin{tikzpicture}[scale=2.5]
\draw[ultra thick,red] (-2,0)--(0,-2) ;
\draw[dashed] (0,-1.6) -- (0,-2);
\draw[dashed] (-0.4,0) -- (-0.4,-1.6);
\draw[dashed] (-0.8,0) -- (-0.8,-1.2);
\draw[dashed] (-1.2,0) -- (-1.2,-0.8);
\draw[dashed] (-1.6,0) -- (-1.6,-0.4);
\draw[dashed] (-1.6,-0.4) -- (0,-0.4);
\draw[dashed] (-1.2,-0.8) -- (0,-0.8);
\draw[dashed] (-0.8,-1.2) -- (0,-1.2);
\draw[dashed] (-0.4,-1.6) -- (0,-1.6);
\draw[dashed,thin,darkspringgreen,pattern color=darkspringgreen,pattern=north east lines] (-2.,0.) -- (2,0.) -- (2,-0.4) -- (-1.6,-0.4)--(-2,0) (2.4,-0.2) node{$p_1=\frac{19}{2}$};
\draw[dashed,blue,pattern color=blue,pattern=north east lines] (-1.6,-0.4) -- (1.6,-0.4) -- (1.6,-0.8) -- (-1.2,-0.8) --  (-1.6,-0.4) (2,-0.6) node{$p_2=\frac{15}{2}$};
\draw[dashed,red,pattern color=red,pattern=north east lines] (-1.2,-0.8) -- (1.2,-0.8) -- (1.2,-1.2) -- (-0.8,-1.2)  (1.6,-1.) node{$p_3=\frac{11}{2}$};
\draw[dashed,gray,pattern color=gray,pattern=north east lines] (-0.8,-1.2) -- (1.2,-1.2) -- (1.2,-1.6) -- (-0.4,-1.6)  (1.6,-1.4) node{$p_4=\frac{9}{2}$};
\draw[dashed,orange,pattern color=orange,pattern=north east lines] (-0.4,-1.6) -- (0,-1.6) -- (0,-2)   (0.4,-1.8) node{$p_5=\frac{1}{2}$};
\draw[clip](0,0) -- ++(2,0) -- ++(0,-0.4)-- ++(-0.4,0)-- ++(0,-0.4)-- ++(-0.4,0)-- ++(0,-0.8)--(0,-1.6) --cycle;
\draw[clip](0,0) -- ++(2,0) -- ++(0,-0.4)-- ++(-0.4,0)-- ++(0,-0.4)-- ++(-0.4,0)-- ++(0,-0.8)--(0,-1.6) --cycle;
\draw[scale=0.4,ultra thick](0,-1.6*3) grid (5,5);
\end{tikzpicture}
\caption{Extended Young diagram corresponding to the state
 $\ph_{-\frac{19}{2}}\ph_{-\frac{15}{2}}\ph_{-\frac{11}{2}}\ph_{-\frac{9}{2}}\ph_{-\frac{1}{2}} |0\rangle $
 of fermionic charge $n=5$. The areas determining the $p_i$'s and $q_i$'s are computed
 by extending the Young diagram by a right triangle of base and height 
 equal to the fermionic charge.  \\ \\ \\ \\}
 \label{fig:Charge5}
\end{subfigure}
\end{center}
\caption{}
\label{fig:charged_young_diag}
\end{figure}


We can thus identify a state either by  the $N$-tuple of Maya diagrams $\vec{\sfm}=(m_1, \dots, m_N)$, 
or the set of half-integers $\{p_{\alpha,i},q_{\alpha,i}\}_{1 \le \alpha \le N}$, or the $N$-tuple of charged Young diagrams 
$\vec{\sfY}_{\nb} := \{(Y_1,n_1), \dots, (Y_N, n_N)\}$:
\begin{equation}
\vec{\sfm} \sim \{p_{\alpha,i},q_{\alpha,j}\}\sim \vec{\sfY}_\nb.
\end{equation}

To write $\tau_W[J]$ as a  fermionic VEV, consider the Fourier expansions of the kernels $\sfa(z,w)$ and $\sfd(z,w)$ in \eqref{eq:adKers} 
and define the ``dressed'' states
\begin{equation}
|\sfd\rangle:=e^{-\hat{\psi}\cdot\sfd\cdot\bar{\hat{\psi}}}|0\rangle,
\label{d_state}
\end{equation}
and
\begin{equation}
\langle\sfa|:=\langle 0|e^{-\bar{\hat{\psi}}\cdot\sfa\cdot\hat{\psi}}, 
\label{a_state}
\end{equation}
where the abbreviated notations $\bar{\hat{\psi}}\cdot\sfa\cdot\hat{\psi}$ and $\hat{\psi}\cdot\sfd\cdot\bar{\hat{\psi}}$ denote
\bea
\bar{\hat{\psi}}\cdot\sfa\cdot\hat{\psi} := \sum_{\alpha,\beta=1}^N\sum_{p,q\in\Zbb'_+}\bar{\hat{\psi}}_{-q}^\alpha\left(\sfa^{p}_{-q}\right)_{\alpha\beta}\hat{\psi}_{-p}^{\beta} \\
\hat{\psi}\cdot\sfd\cdot\bar{\hat{\psi}} := \sum_{\alpha,\beta=1}^N\sum_{p,q\in\Zbb'_+}\hat{\psi}_{-q}^\alpha\left(\sfd_{p}^{-q}\right)_{\alpha\beta}\bar{\hat{\psi}}_{-p}^{\beta} .
\label{psi_d_psi}
\eea
We then have:
\begin{proposition}
\label{ad_scalar_prod}
The Widom $\tau$-function admits the following representation as a product of free fermion states (see \cite{Gavrylenko2016a}):
\begin{equation}
\label{eq:TauWF}
\tau_W[J]=\langle\sfa|\sfd\rangle,
\end{equation}
and admits the following combinatorial expansion as a sum of products of finite determinants
\begin{equation} 
\tau_W=\sum_{\#(\vec{\textbf{p}})=\#(\vec{\textbf{q}})}(-)^{\#(\vec{\textbf{p}})} \det\left(\sfa^{\vec{\textbf{p}}}_{-\vec{\textbf{q}}} \right)\det\left(\sfd^{-\vec{\textbf{q}}}_{\vec{\textbf{p}}} \right).
\label{eq:WidomMinor}
\end{equation}
 where 
 \be
(\vec{\pb}, \vec{\qb}) = \{p_{\alpha,i},q_{\alpha, j}  \}_{\alpha =1, \dots, N},  \quad  \  q_{\alpha,i}, p_{\alpha,j} > 0.
 \ee

\end{proposition}
\begin{proof}
First,  insert a sum over a complete set of intermediate states:
\begin{equation}
\langle\sfa|\sfd\rangle=\sum_{\vec{\sfm}}\langle\sfa|\vec{\sfm}\rangle\langle\vec{\sfm}|\sfd\rangle,
\label{intermed_state_sum}
\end{equation}
and then note that
\begin{equation}
\langle\vec{\sfm}|\sfd\rangle=\langle0|\prod_{\alpha=1}^N\left(\prod_{i=1}^{\#\{p_{\alpha,i}\}}\bar{\hat{\psi}}_{-p_{\alpha,i}}^\alpha\prod_{j=1}^{\#\{q_{\alpha,j}\}}\hat{\psi}_{-q_{\alpha,j}}^\alpha \right)\sum_{n=0}^\infty\frac{(-\hat{\psi}\cdot\sfd\cdot\bar{\hat{\psi}})^n}{n!}|0\rangle.
\end{equation}
Using Wick's theorem to contract all the fermions, we obtain
\begin{equation}
\langle\vec{\sfm}|\sfd\rangle=(-)^{\#\{p_{i,\beta}\}}\det\left(\sfd^{-\vec{\textbf{q}}}_{\vec{\textbf{p}}}\right).
\end{equation}
and, similarly,
\begin{equation}
\langle\sfa|\vec{\sfm}\rangle=\det\left(\sfa^{\vec{\textbf{p}}}_{\vec{-\textbf{q}}}\right),
\end{equation}
 labelled in terms of N-tuples $\vec{\sfm}$ of Maya diagrams $\{m_\alpha=\{p_{\alpha,i},q_{\alpha, j} \} \}_{\alpha = 1, \dots,N}$.
 Substituting these in (\ref{intermed_state_sum}) gives (\ref{eq:WidomMinor}), which is
 the minor expansion of the Widom $\tau$-function from \cite{Cafasso2017}. 
 \end{proof}


\subsection{$N$-component orthogonal fermions and $L_SS O(N)$ $\tau$-functions}
\label{sec:OrtFerm}
\label{[orthog_fermions_O_N_tau}
If the Riemann Hilbert problem takes values in an orthogonal loop group, with jump matrix and solutions satisfying \eqref{eq:ORHP}, the Widom $\tau$-function is not the 
object of fundamental interest, since  the determinant \eqref{eq:WidomTauDef} then turns out to be the square of a Pfaffian.
This is best seen by introducing a new  basis for the fermionic creation and annihilation operators, which we call {\em orthogonal fermions}.
\begin{definition}[Orthogonal Fermions]
Depending on the choice of matrix involution $S$ entering in the quadratic form \eqref{eq:PairingO},
we define (S-)orthogonal fermion fields, $\{\chi^\alpha(z;S),\hat{\chi}^\alpha(z;S)\}$  in terms charged fermion fields as
\begin{align}\label{eq:ofField}
\chi^{\alpha}(z;S):=\frac{1}{\sqrt{2}}\left(\ph^{\alpha}(z)+(S\phb)^{\alpha}(z)\right), && \hat{\chi}^{\alpha}(z;S):=\frac{i}{\sqrt{2}}\left(\ph^{\alpha}(z)-(S\phb)^{\alpha}(z)\right),
\end{align}
where $(S\phb)^{\alpha}(z)$ denotes the $\alpha$ component of the  product of the matrix $S$ with the column vector $\phb(z)$ 
whose components are $\phb^{\alpha}(z)$.
\end{definition}

In terms of their Fourier expansion coefficients
\be
\chi^{\alpha}(z;S)=\sum_{p\in\mathbb{Z}'}\chi^{\alpha}_{p}(S)z^{-p-\frac{1}{2}}, \quad \hat{\chi}^{\alpha}(z;S)=\sum_{p\in\mathbb{Z}'}\chi^{\alpha}_{p}(S)z^{-p-\frac{1}{2}},
\ee
\eqref{eq:ofField} is equivalent to
\be
\label{ofModes}
\chi^{\alpha}_{p}(S)=\frac{1}{\sqrt{2}}\left(\ph^\alpha_p+(S\phb)^\alpha_p \right), \quad
\hat{\chi}^{\alpha}_{p}(S)=\frac{i}{\sqrt{2}}\left(\ph^\alpha_p-(S\phb)^\alpha_p \right).
\ee
The anticommutators between the orthogonal fermions follow from those \eqref{eq:cfanticom} for  charged fermions
\be
\label{eq:chianticom}
[\chi^{\alpha}_{p}(S),\chi^{\beta}_{q}(S) ]_+
=[\hat{\chi}^{\alpha}_{p}(S),\hat{\chi}^{\beta}_{q}(S) ]_+
=S^{\alpha\beta}\delta_{p,-q}, 
\quad [\chi^{\alpha}_{p}(S),\hat{\chi}^{\beta}_{q}(S) ]_+=0,
\ee
which generates the Clifford algebra corresponding to the quadratic form \eqref{eq:PairingO} on the 
space spanned by $\{\chi^{\alpha}_{p}(S), \hat{\chi}^{\alpha}_{p}(S)\}$.
The vacuum annihilation conditions also follow from those for charged fermions:
\be
\label{eq:chivacuum}
\chi^{\alpha}_{p}(S)|0\rangle=\hat{\chi}^{\alpha}_{p}(S)|0\rangle=0, \quad p>0.
\ee
\begin{remark}
The simplest case is $S=\mathbb{I}$. Equation \eqref{eq:ofField} is then the change of variables between (the chiral part of) a Dirac fermion, 
corresponding to the field $\hat{\psi}(z),\bar{\hat{\psi}}(z)$, and two Weyl-Majorana fermions, known as real fermions in the physical literature, 
given by the fields $\chi(z;\mathbb{I}),\hat{\chi}(z;\mathbb{I})$.
\end{remark}
In the following, to simplify notation we will drop the explicit dependence of fermions on the matrix $S$ if not needed. Since the two species of 
orthogonal fermions $\chi,\,\hat{\chi}$ anticommute, we can define two smaller Fock spaces spanned by application of  $\chi$ or $\hat{\chi}$
to $|0\rangle$. In view of the anticommutation relations \eqref{eq:chianticom} and the vacuum condition \eqref{eq:chivacuum}, the basis states 
in these smaller spaces are labeled by $N$ strictly decreasing sequences of positive half-integers 
\be
p_{\alpha,1}>\dots>p_{\alpha,n_\alpha} \ge 1/2 \quad\alpha=1,\dots N
\ee
or equivalently, by an $N$-tuple of strict partitions (denoted $SP$), including a possible $0$ part:
\be
\label{eq:stricthi}
    \vec{\lambda}:=\left(\lambda^{(1)},\dots,\lambda^{(N)} \right)\in SP, \quad \lambda^{(\alpha)}:=\left(p_{\alpha,1}-\frac{1}{2},\dots,p_{\alpha,n_{\alpha}}-\frac{1}{2}\right),
\ee 
where $n_\alpha :=\#(\lambda^{(\alpha)})$ is now  the cardinality of the strict partition $\lambda^{(\alpha)}$.
We denote by
\be
\#(\vec{\lambda}) := \sum_{\alpha=1}^N n_\alpha
\ee
the total cardinality of the $N$-tuple of strict partitions $\vec{\lambda}$. 
The basis states are then denoted
\begin{equation}\label{eq:StateReal}
|\vec{\lambda}\rangle=\overrightarrow{\prod_{\alpha=1}^N}\chi^{\alpha}_{-p_{\alpha,1}}\dots\chi^{\alpha}_{-p_{\alpha,n_\alpha}}|0\rangle,
\end{equation}
with an identical formula for the hatted fermions $\hat{\chi}^{\alpha}_{-p_{\alpha,i}}$.
 The arrow over the product symbol means that the ordering is such that the  $\alpha$'s are increasing to the right.
 For example, for $N=2$, with strict partitions $\lambda_1=(2,0),\,\lambda_2=(1)$, the corresponding state will be
\begin{equation}
|\vec{\lambda}\rangle=|(2,0),\,(1)\rangle=\chi^{1}_{-\frac{5}{2}}\chi^{1}_{-\frac{1}{2}}\chi^{2}_{-\frac{3}{2}}|0\rangle.
\end{equation}

Define the ``dressed'' states
\bea
|\sfd_O\rangle &\&:= e^{-\frac{1}{2}\chi\cdot\sfd S\cdot\chi}|0\rangle,  \quad |\sfd_{\hat{O}}\rangle:= e^{-\frac{1}{2}\hat{\chi}\cdot\sfd S\cdot\hat{\chi}}|0\rangle,
\label{eq:aO}
\\
\langle\sfa_O| &\&:=\langle 0| e^{-\frac{1}{2}\chi\cdot S\sfa\cdot\chi}, \quad \langle\sfa_{\hat{O}}|:= \langle 0|e^{-\frac{1}{2}\hat{\chi}\cdot S\sfa\cdot\hat{\chi}},
\label{eq:dO}
\eea
where
\bea
\chi\cdot S \sfa \cdot \chi &\&:= \sum_{\alpha,\beta=1}^N\sum_{p,q\in\Zbb'_+}\chi^\alpha_p\left(S\sfa^{p}_{-q}\right)_{\alpha\beta}\chi^\beta_q, \qquad
\hat{\chi}\cdot S\sfa \cdot \hat{\chi} := \sum_{\alpha,\beta=1}^N\sum_{p,q\in\Zbb'_+}\hat{\chi}^\alpha_p\left( S\sfa^{p}_{-q} \right)_{\alpha\beta}\hat{\chi}^\beta_q, \cr
&\& \\
\chi \cdot \sfd S \cdot \chi &\&:=\sum_{\alpha,\beta=1}^N \sum_{p,q\in\Zbb'_+} \chi^\alpha_{-q}\left(\sfd_{p}^{-q}S_{\alpha\beta} \chi^\beta_{-p} \right) , \qquad
\hat{\chi}\cdot \sfd S \cdot\hat{\chi} :=\sum_{\alpha,\beta=1}^N \sum_{p,q\in\Zbb'_+} \hat{\chi}^\alpha_{-q}\left(\sfd_{p}^{-q}S_{\alpha\beta} \hat{\chi}^\beta_{-p} \right)\cr
&\&
\eea

\begin{theorem}
\label{thm:RealFTau}
Let $\tau_W[J]$ be the Widom $\tau$-function \eqref{eq:WidomTauDef}. If the jump $J$ in the Riemann-Hilbert factorization (\ref{eq:JumpRHP}) satisfies 
the orthogonal reduction condition (\ref{eq:ORHP}), then $\tau_W[J]   $ can be written as a square
\begin{equation}
\tau_W[J]=\tau_O[J]^2,
\label{tau_widom_tau_pfaff}
\end{equation}
where 
\be
\tau_O[J]:=\langle \sfa_O|\sfd_O\rangle = \langle \sfa_
{\hat{O}}|\sfd_{\hat{O}}\rangle.
\label{eq:TauD}
\ee
\end{theorem}

\begin{proof}
We first make the change of fermionic basis \eqref{ofModes} in the expression \eqref{eq:TauWF} for the $\tau$-function, using $S^{-1}=S$:
\bea
\sum_{\alpha,\beta}\sum_{p,q>0}(d^{-q}_{p})_{\alpha\beta}\hat{\psi}^\alpha_{-q}\overline{\hat{\psi}}^\beta_{-p} &\& =\frac{1}{2}\sum_{\alpha,\beta}\sum_{p,q>0}(d^{-q}_{p})_{\alpha\beta}\left(\chi^\alpha_{-q}+i\hat{\chi}^\alpha_{-q}\right)\left((S\chi)^\beta_{-p}-i(S\hat{\chi})^\beta_{-p}\right) \cr
&\& =\frac{1}{2}\sum_{\alpha,\beta}\sum_{p,q>0}(d^{-q}_{p}S)_{\alpha\beta}\left(\chi^\alpha_{-q}\chi^\beta_{-p}+\hat{\chi}^\alpha_{-q}\hat{\chi}^\beta_{-p} \right) \cr
&\& -\frac{i}{2}\sum_{\alpha,\beta}\sum_{p,q>0}(d^{-q}_{p}S)_{\alpha\beta}\left(\chi^\alpha_{-q}\hat{\chi}^\beta_{-p}-\hat{\chi}^\alpha_{-q}\chi^\beta_{-p} \right).
\eea
The last line vanishes due to the antisymmetry \eqref{eq:antisymB2} of $\sfd$, stemming from the orthogonality condition (\ref{eq:ORHP}), since the combination
\be
\chi^\alpha_{-q}\hat{\chi}^\beta_{-p}-\hat{\chi}^\alpha_{-q}\chi^\beta_{-p} =\chi^\alpha_{-q}\hat{\chi}^\beta_{-p}+\chi^\beta_{-p}\hat{\chi}^\alpha_{-q} 
\ee
is symmetric under the exchange $(q,\alpha)\leftrightarrow(p,\beta)$. By eq.~(\ref{psi_d_psi}), we are therefore left with
\be
\hat{\psi}\cdot\sfd\cdot\overline{\hat{\psi}}=\frac{1}{2}\left(\chi\cdot\sfd S\cdot\chi+\hat{\chi}\cdot\sfd S\cdot\hat{\chi} \right).
\ee
A completely analogous relation holds for $\sfa$, leading to
\be
\overline{\hat{\psi}}\cdot\sfa\cdot\hat{\psi}=\frac{1}{2}\left(\chi\cdot S\sfa\cdot\chi+\hat{\chi}\cdot S\sfa \cdot\hat{\chi} \right).
\ee
We then have
\bea\label{eq:TauOFact}
\tau_W[J] &\& =\langle\sfa|\sfd\rangle=\langle0|e^{-\frac{1}{2}\left(\chi\cdot S\sfa\cdot\chi+\hat{\chi}\cdot S\sfa\cdot\hat{\chi} \right)}e^{-\frac{1}{2}\left(\chi\cdot\sfd S\cdot\chi+\hat{\chi}\cdot\sfd S\cdot\hat{\chi} \right)}|0\rangle\cr
&\&=\langle\sfa_{O}|\sfd_{O}\rangle\langle\sfa_{\hat{O}}|\sfd_{\hat{O}}\rangle=\left(\langle\sfa_{O}|\sfd_{O}\rangle \right)^2 =\left(\tau_O[J]\right)^2,
\eea
where
\bea
\tau_O[J]&\&:=\langle\sfa_{O}|\sfd_{O}\rangle=\langle0|e^{-\frac{1}{2}\chi\cdot S\sfa\cdot\chi}e^{-\frac{1}{2}\chi\cdot\sfd S\cdot\chi}|0\rangle \cr
      &\&=\langle\sfa_{\hat{O}}|\sfd_{\hat{O}}\rangle=\langle0|e^{-\frac{1}{2}\hat{\chi}\cdot S\sfa\cdot\hat{\chi}}e^{-\frac{1}{2}\chi\hat{\cdot}\sfd S\cdot\hat{\chi}}|0\rangle .
\eea
The factorization of expectation values in the second line of \eqref{eq:TauOFact} follows from the fact that 
$\{\chi^\alpha_p\}$ and $\{\hat{\chi}^\alpha_p\}$ are mutually anticommuting sets of fermionic operators, each leading to its own ``restricted'' Fock space. 
(See for example the factorization Lemma 2.2 in \cite{HO1}, and \cite{you1989polynomial} for the analogue of this statement in the BKP hierarchy.)
\end{proof}
The formulation in terms of free fermions leads to the following general form of the combinatorial expansion for this $\tau$-function:
\begin{theorem}
\label{thm:PfMinor}
The $\tau$-function $\tau_O[J]$ has the following combinatorial expansion:
\begin{equation}\label{eq:CombExpD}
\tau_O[J]=\sum_{\vec{\lambda}\in SP}\Pf(S\sfa_{\vec{\lambda}})\Pf(\sfd S_{\vec{\lambda}}),
\end{equation}
where $\sfa_{\vec{\lambda}}, \sfd_{\vec{\lambda}}$ are the $\#(\lambda)\times\#(\lambda)$ submatrices of $\sfa,\sfd$ labeled by the $N$-tuple of strict partition $\vec{\lambda}$. 
\end{theorem}
\begin{proof}
Inserting an intermediate sum over a complete set of states in the fermionic expression \eqref{eq:TauD} for $\tau_O[J]$ gives:
\begin{equation}
\tau_O[J]=\sum_{\vec{\lambda}\in SP}\langle0|e^{-\frac{1}{2}\chi\cdot S\sfa\cdot\chi}|\vec{\lambda}\rangle\langle\vec{\lambda}|e^{-\frac{1}{2}\chi\cdot\sfd S\cdot\chi}|0\rangle.
\end{equation}
We now show that
\begin{equation}
\langle\vec{\lambda}|e^{-\frac{1}{2}\chi\cdot\sfd S\cdot\chi}|0\rangle=\langle0|\prod_{i=1}^{|p_{\alpha,i}|}\chi^\alpha_{p_{\alpha,i}}\sum_{n=0}^{\infty}\frac{1}{n!}\left(-\frac{1}{2}\chi\cdot\sfd S\cdot\chi \right)^n|0\rangle=(-1)^{\frac{\#(\lambda)}{2}}\Pf(\sfd_{\vec{\lambda}}),
\end{equation}
where $\sfd_{\vec{\lambda}}$ (which is necessarily of even dimension) is the square submatrix identified by the strict $N$-tuple of partitions $\vec{\lambda}$. The Pfaffian factor comes from the Wick contractions of all the fermions, while the sign comes from the exponential. This can be computed from the Grassmann algebra valued (Berezinian) Gaussian integral
\begin{equation}
\begin{split}
\langle&\vec{\lambda}|e^{-\frac{1}{2}\chi\cdot\sfd\cdot\chi}|0\rangle  =\int\overleftarrow{\prod_{\alpha=1}^{N}}\left[d\chi_{p_{\alpha,n_{\alpha}}}^{\alpha}\dots d\chi_{p_{\alpha,1}}^{\alpha}\right]e^{-\frac{1}{2}\chi\cdot\sfd S\cdot\chi}\\
& =\sum_{k=0}^{\infty}\frac{1}{k!}\left(-\frac{1}{2} \right)^k\int\overleftarrow{\prod_{\alpha=1}^{N}}\left[d\chi_{p_{\alpha,n_{\alpha}}}^{\alpha}\dots d\chi_{p_{\alpha,1}}^{\alpha}\right]\left(\sum_{\alpha\beta=1}^{N}\sum_{p,q>0}\left(\sfd^{-p}_qS \right)_{\alpha\beta}\chi^\alpha_{-p}\chi^\beta_{-q} \right)^k.
\end{split}
\end{equation}
In this sum, due to the rules of Berezinian  integration, the only terms that are non-vanishing are those of order $\ell:=\frac{\#(\lambda)}{2}$, where all the $\#(\lambda)$ fermions in the measure are saturated. Furthermore, each such term appears with the sign of the permutation $\pi$ that brings 
the fermionic insertions into the correct order. This means that
\begin{equation}
\begin{split}
\langle\vec{\lambda}|e^{-\frac{1}{2}\chi\cdot\sfd\cdot\chi}|0\rangle & =\frac{(-1)^{\ell}}{2^{\ell}\ell!}\sum_{\pi\text{ perm. of }\vec{\lambda}}\mathrm{sgn}(\pi)\prod_{\alpha,\beta=1}^N\left(\sfd^{-\pi(p_{\alpha,1})}_{\pi(q_{\beta,1})}S \right)_{\pi(\alpha)\pi(\beta)}\dots\left(\sfd^{-\pi\left(p_{\alpha,\ell}\right)}_{\pi\left(q_{\beta,\ell}\right)}S \right)_{\pi(\alpha)\pi(\beta)}\\
& =(-1)^\ell\Pf(\sfd S)_{\vec{\lambda}}.
\end{split}
\end{equation}
Essentially the same computation shows that 
\be
\langle0|e^{-\frac{1}{2}\chi\cdot S\sfa\cdot\chi}|\vec{\lambda}\rangle=(-1)^\ell\Pf(S\sfa)_{\vec{\lambda}}.
\ee
Combining the two, the terms of opposite sign cancel, and we have
\begin{equation}
\tau_O[J]=\sum_{\vec{\lambda}\in (SP)^N}\Pf(S\sfa)_{\vec{\lambda}}\Pf(\sfd S)_{\vec{\lambda}}.
\end{equation}
\end{proof}

We see that the Fredholm determinant Widom $\tau$-function $\tau_W[J]$ for orthogonal loop group elements 
is the square of a $\tau$-function $\tau_0[J]$, which admits an expansion in Pfaffian minors.
It is  natural to expect that $\tau_O[J]$ itself  is the Pfaffian of a $2$-form on the Hilbert space $\HH^N$. 
This is most easily seen by giving an operatorial interpretation of the above fermionic construction, as follows. 
We can conjugate the kernel inside the Fredholm determinant to get:
\begin{equation}
\begin{split}
\tau_W[J] & =\det_{\HH^N}\left( \begin{array}{cc}
\mathbb{I}_+ & \sfa \\
\sfd & \mathbb{I}_-
\end{array} \right)=\det_{\HH^N}\left( \begin{array}{cc}
\mathbb{S} & 0 \\
0 & \mathbb{I}_- 
\end{array} \right)\left( \begin{array}{cc}
\mathbb{I}_+ & \sfa \\
\sfd & \mathbb{I}_-
\end{array} \right)\left( \begin{array}{cc}
\mathbb{S} & 0 \\
0 & \mathbb{I}_- 
\end{array} \right)\\
&=\det_{\HH^N}\left( \begin{array}{cc}
\mathbb{I}_+ & S\sfa \\
\sfd S & \mathbb{I}_-
\end{array} \right)=\det_{\HH^N}\left( \begin{array}{cc}
\mathbb{I}_+ & S\sfa \\
\sfd S & \mathbb{I}_-
\end{array} \right),
\end{split}
\end{equation}
where $\mathbb{I}_\pm:=\mathbb{I}|_{\HH^N_\pm}$, and $\mathbb{S}$ acts as a matrix $S$ on the $\mathbb{C}^N$ component of $\HH^N=L^2(S^1)\otimes\mathbb{C}^N$ 
and as the identity operator on $L^2(S^1)$. To write $\tau_W[J]$ as the square of a Pfaffian, introduce the operator $\Omega:\HH^N_{\pm}\rightarrow \HH^N_{\mp}$ acting on the monomial basis as follows:
\begin{equation}
\Omega\left(z^{p-\frac{1}{2}} \right)= \begin{cases}
-z^{-p-\frac{1}{2}}, & p>0, \\
z^{-p-\frac{1}{2}}, & p<0.
\end{cases}
\end{equation}
Then
\begin{equation}
\tau_W[J]=\det_{\HH^N}\left(\mathbb{I}+ \left(\begin{array}{cc}
0 & S\sfa \\
\sfd S & 0
\end{array}\right) \right)
=\det_{\HH^N}\left(\mathbb{I}+ \left(\begin{array}{cc}
S\sfa\Omega_+ & 0 \\
0 & \sfd S\Omega_-
\end{array} \right)\Omega^{-1} \right),
\end{equation}
where $\Omega_\pm=\Omega|_{\HH^N_{\pm}}$. Using a block operator version of the Fredholm Pfaffian definition from \cite{rains2000correlation},  we have
\begin{equation}
\det_{\HH^N}\left(\mathbb{I}+ \left(\begin{array}{cc}
S\sfa\Omega_+ & 0 \\
0 & \sfd S\Omega_-
\end{array} \right)\Omega^{-1} \right)=\Pf\left(\Omega+\left(\begin{array}{cc}
S\sfa\Omega_+ & 0 \\
0 & \sfd S\Omega_-
\end{array} \right) \right)^2,
\end{equation}
which gives the identification
\begin{equation}
\label{eq:TauFredPf}
\tau_O[J]=\pm\Pf\left(\Omega+\left(\begin{array}{cc}
S\sfa\Omega_+ & 0 \\
0 & \sfd S\Omega_-
\end{array} \right) \right),
\end{equation}
the sign being chosen in such a way that the leading term in $\tau_O[J]$ is $+1$. In 
\eqref{eq:TauFredPf} we are, strictly speaking, making a slight abuse of notation: a Pfaffian is well-defined only for $2$-forms, while $\sfa,\sfd,\Omega$ were introduced as endomorphisms of $\HH^N$, i.e. $(1,1)$ tensors. We can however use the isomorphism \eqref{eq:DualPair} between $\HH^N$ and $\HH^{N*}$ to view an endomorphism with antisymmetric Fourier coefficients as a $2$-form. Under this identification, 
\be
\Omega \in \HH^N_+\wedge\HH^N_+ \oplus \HH^N_-\wedge\HH^N_-, \qquad
S\sfa\Omega_+  \in\HH^N_+\wedge\HH^N_+, \qquad \sfd S\Omega_-\in \HH^N_-\wedge \HH^N_-.
\ee 
This  also explains the geometric reason behind the appearance of the matrix $S$ multiplying our operators; to turn $\sfa$ and $\sfd$ into $2$-forms, we have to ``lower an index''  with the quadratic form \eqref{eq:PairingO}, while the factor $\Omega_\pm$ assures that the expression for $\tau_O[J]$, thought of as an operator, is still consistent with the $\HH^N_\pm$ splitting\footnote{$\tau_O[J]$ can also be identified with the relative Pfaffian of the operators $S\sfa$ and $\sfd S$, which was introduced in \cite{jaffe1989Pfaffians}. The definition of the relative Pfaffian is essentially the expansion \eqref{eq:CombExpD}.}.

Figure \ref{Fig:MinorPf} illustrates the $S\sfa$-Pfaffian minor corresponding to the case of the pairs of strict
partitions $\vec{\lambda}_1=\left((1,0),\emptyset \right)$ and $\vec{\lambda}_2=\left((0),(1) \right)$.

\begin{figure}[h]
\begin{center}
\begin{subfigure}[t]{.45\textwidth}
\centering
\includegraphics[width=\textwidth]{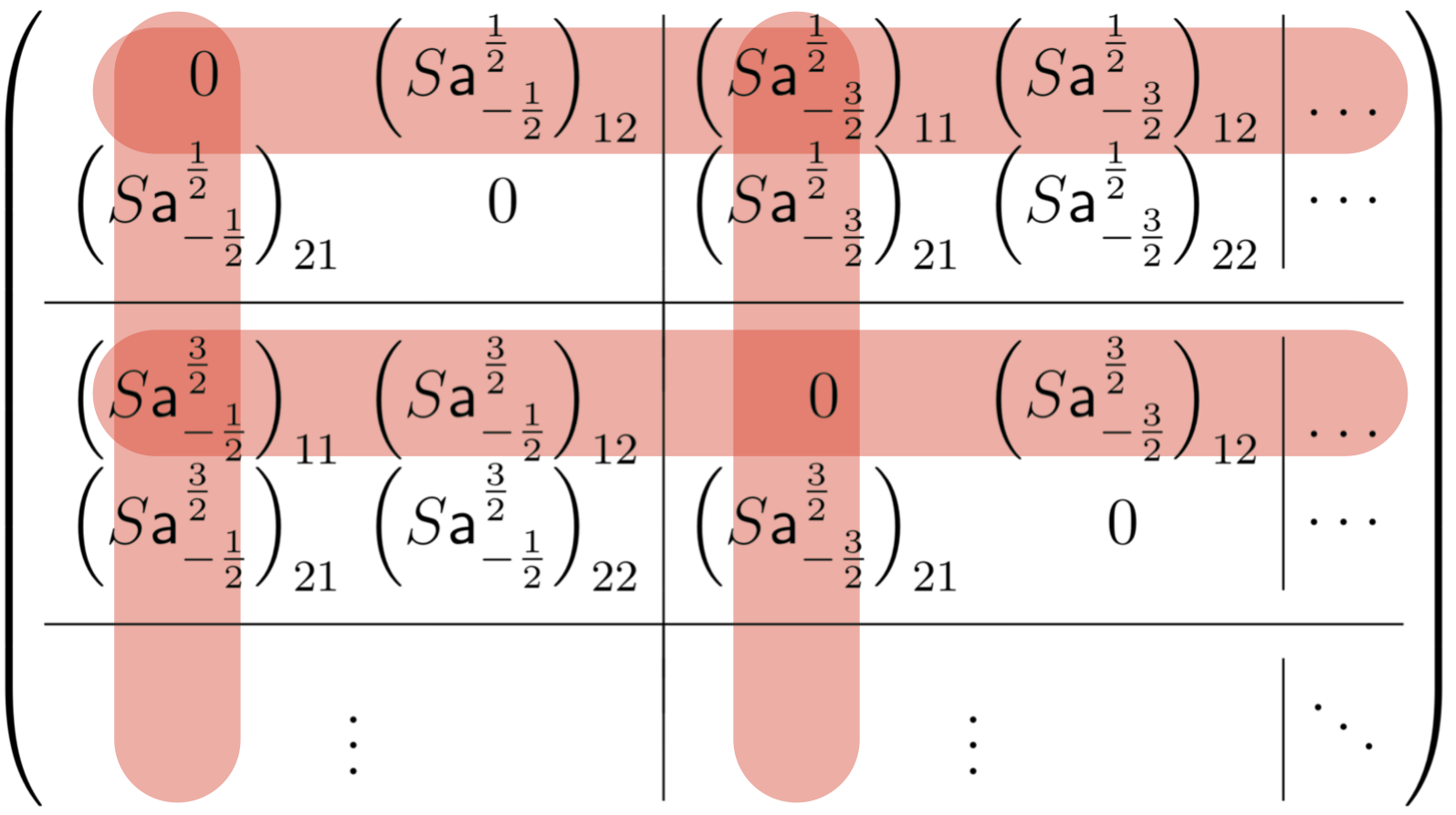}
\caption{{\footnotesize  Pfaffian minor of the antisymmetric $2\times 2$ block infinite matrix $S\sfa$ associated to the pair of strict partitions
 $\vec{\lambda}_1=\left((1,0),\emptyset \right)$.
\\$\Pf\left(S\sfa\right)_{\vec{\lambda}_1}=\Pf\left( \begin{array}{cc}
0 & \left(S\sfa^{\frac{1}{2}}_{-\frac{3}{2}} \right)_{11} \\
\left(S\sfa^{\frac{3}{2}}_{-\frac{1}{2}} \right)_{11} & 0
\end{array} \right) $.}}
\end{subfigure}\hfill
\begin{subfigure}[t]{.45\textwidth}
\centering
\includegraphics[width=\textwidth]{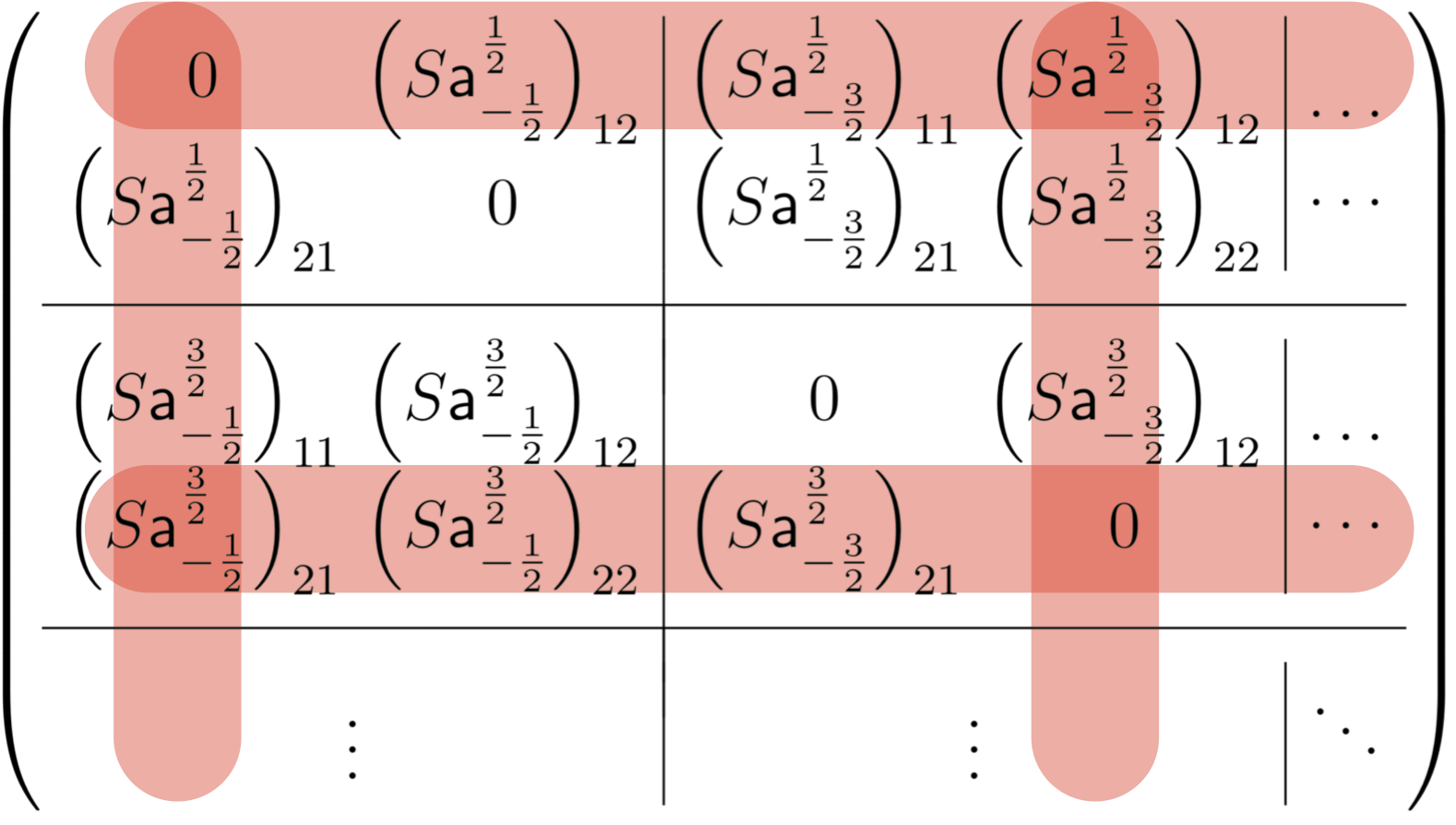}
\caption{{\footnotesize Pfaffian minor of the antisymmetric $2\times 2$ block infinite matrix $S\sfa$, associated to the pair of strict partitions
 $\vec{\lambda}_2=\left((0),(1) \right)$. \\
\\$\Pf\left(S\sfa\right)_{\vec{\lambda}_2}=\Pf\left( \begin{array}{cc}
0 & \left(S\sfa^{\frac{1}{2}}_{-\frac{3}{2}} \right)_{12} \\
\left(S\sfa^{\frac{3}{2}}_{-\frac{1}{2}} \right)_{21} & 0
\end{array} \right) $.}}
\end{subfigure}
\caption{$2\times 2$ Pfaffian minors of the infinite antisymmetric block matrix $S\sfa$. The strict partitions label the rows and columns.}\label{Fig:MinorPf}
\end{center}
\end{figure}
\bigskip


\section{Pfaffian $\tau$-function and the orthogonal Drinfeld-Sokolov hierarchy}
\label{drinfeld_sokolov_pfaff}
We turn to a first application of the results of Section \ref{sec:OrtFerm} and  our general formalism. This will allow us to write the $\tau$-function for the 
Drinfeld-Sokolov hierarchy \cite{drinfeld_sokolov1984} as a Pfaffian, as well as providing an expansion in terms of Pfaffian minors labeled by $N$-tuples of strict partitions, where $N=2\ell+1$ for the $B_\ell^{(1)}$ case and $N=2\ell$ for the $D_\ell^{(1)}$ case.

First recall some standard notation for affine Kac-Moody algebras. Let $\ell$ be the rank of the Lie algebra $\mathfrak{g}=\mathfrak{so}_N$, where either
 $N=2\ell+1$, for $B_\ell$ or $N=2\ell$ for $D_\ell$,  let $\mathfrak{h}\subset\mathfrak{so}_N$ be a Cartan subalgebra, and $\Pi:=\{\alpha_1,\dots,\alpha_\ell\}$ a set of simple roots. If $\Delta\subset\mathfrak{h}^*$ is the root system, we have the decomposition 
\begin{equation}
\mathfrak{g}=\mathfrak{h}\bigoplus_{\alpha\in\Delta}\mathfrak{g}_\alpha.
\end{equation}
For $\alpha\in\Delta$, let $h_{\alpha}\in\mathfrak{h}$ be defined by
\begin{equation}
B\left(h_\alpha ,X \right)=\alpha(X),\quad \forall \ X\in\mathfrak{h},
\end{equation}
where $B$ is the Killing form, and
\begin{equation}
H_\alpha:=\frac{2h_\alpha}{B(h_\alpha,h_\alpha)}.
\end{equation}
Define the $\ad^*$ and Weyl group invariant bilinear form $( \, \cdot  \, | \,  \cdot  \,)$ on  $\grh^*$ by 
\begin{equation}
\left(\alpha|\beta \right):= \kappa\, B\left(H_\alpha,H_\beta \right), \quad\forall \ \alpha,\beta\in\Delta,
\end{equation} 
where $\kappa$ is determined by requiring
\be
\left(\theta|\theta \right)=2
\ee
for the highest root  $\theta$ with respect to $\Pi$.
Let
\be
E_i \in\mathfrak{g}_{\alpha_i}, \quad F_i\in\mathfrak{g_{-\alpha_i}}, \quad H_i  :=\frac{2H_{\alpha_i}}{\left(\alpha_i|\alpha_j \right)},\\
\ee
be a set of Weyl generators of $\mathfrak{g}$, satisfying
\be
 [H_i,E_j] =A_{ij}E_j, \quad [H_i,F_j]=-A_{ij}F_j, \quad
 [E_i,F_j] =H_i\delta_{ij}, \quad  i, j =1,\dots,\ell , 
\ee
where $A_{ij}$ is the Cartan matrix. Define the Chevalley involution $\SS$ by
\begin{align}
\SS(E_i)=-F_i, \quad \SS(F_i)=-E_i, \quad \SS(H_i)=-H_i, \quad i=1,\dots,\ell,
\end{align}
extended as an involutive automorphism of the algebra. Fix $E_{-\theta}\in\mathfrak{g}_{-\theta}$, $E_{\theta}\in\mathfrak{g}_\theta $ by the conditions
\be
\left(E_\theta|E_{-\theta} \right)=1, \quad \SS (E_{-\theta})=-E_\theta,
\ee
and define
\be
E_0:=E_{-\theta}, \quad F_0:=E_{\theta}, \quad H_0:=[E_0,F_0].
\ee
A set of Weyl generators for the affine Kac-Moody algebra $\grg^{(1)}$, obained as a central extension of the loop algebra $L\grg$
 is\footnote{More generally, one could set  $h_i:= H_i+\delta_{i,0}c$, where $c$ is the central extension element.
But since the central extension will play no part, we omit it, and just deal with the loop algbra $L\grg$.}
\be
e_i:=z^{\delta_{i,0}}E_i, \quad f_i:=z^{-\delta_{i,0}}F_i, \quad h_i:= H_i.
\ee

We use the matrix realization of $L_S\grso(N)$ introduced in \cite{drinfeld_sokolov1984, kac1990infinite, Cafasso2019}, 
which is recalled in Appendix \ref{sec:MatRep}. In this representation, elements of the Lie algebra  satisfy 
\begin{equation}\label{eq:ChevAntis}
X^t=-SXS=-\mathcal{S}(X),
\end{equation}
where $S$ is the matrix representation of the Chevalley involution, acting by conjugation, defined in  eqs.~\eqref{eq:ChevalleyB} and \eqref{eq:ChevalleyD}.) 
We introduce the shift matrix
\be
\Lambda:= \sum_{i=1}^nE_i+zE_{-\theta}\in L_+\mathfrak{g},
\ee
satisfying
\be
\Lambda^{1+h}=z\Lambda,
\ee
where 
\be
h= \begin{cases}
2\ell=N-1, & \ \text{for } \ B_\ell^{(1)}, \\
2\ell-2=N-2, & \ \text{for } \ D_\ell^{(1)}
\end{cases}
\ee
is the Coxeter number. The time evolution of the corresponding Drinfeld-Sokolov hierarchy is encoded in $\Psi_+(z)\in L_{S+} SO(N)$, defined by:
 \footnote{In  general, the $D_\ell^{(1)}$ case admits two sets of abelian flows, labeled by times $t_{2k+1}$ and $t'_{2k+1}$, as in  \cite{Cafasso2019}. For the sake of clarity of exposition, we restrict  ourselves to the first only, and set $t'_{2k+1}=0$. }
\be
\label{eq:PsipB}
\Psi_+(z):=e^{Y(z, \tb_0)}, \quad Y(z, \tb_0):= \left\{\sum_{k=1}^{\infty}t_{2k+1}\Lambda^{2k+1}\right\}.
\ee
where
\be
\tb_0:= (t_1, t_3, t_5, \dots).
\ee
The initial condition  is encoded in an arbitrary negative element $\Psi_-\in L_{S-}SO(N)$, which we write as 
\be
\label{eq:PsimDS}
\Psi_-(z):= e^{-X(z)}, \quad X(z)\in L_{S-}\grso(N).
\ee

In \cite{Cafasso2018} it was shown that the Widom constant \eqref{eq:WidomTauDef} is the square of the $\tau$-function of the Drinfeld-Sokolov
 hierarchy with time evolution defined by \eqref{eq:PsipB} and initial condition specified by \eqref{eq:PsimDS}, namely:
\begin{equation}\label{eq:WDS}
\tau_W[J]=\tau_{DS}^2.
\end{equation}
Equation \eqref{eq:WDS} is already a strong indication that $\tau_{DS}$ should be a Pfaffian, since  it is the square root of the Widom constant $\tau$-function 
$\tau_W[J]$. In fact, this is a special instance of Theorem \eqref{thm:RealFTau}, following from \eqref{eq:ChevAntis}, 
which means that both $\Psi_+,\Psi_-$ satisfy our orthogonal loop group condition \eqref{eq:ORHP}:
\begin{proposition}\label{thm:DSTau}
The $\tau$-function of the Drinfeld-Sokolov hierarchy of type $D_\ell^{(1)}$ or $B_\ell^{(1)}$ is the Pfaffian expression \eqref{eq:TauD}, where $S$ is the Chevalley involution of the appropriate loop algebra, where $\Psi_+(z)$ entering in $\sfa(z,w)$ is given by \eqref{eq:PsipB}, and
$\Psi_-(z)$ by (\ref{eq:PsimDS})  where $X_-$ is an arbitrary element of $L_S \grso(N)$. 
\end{proposition}

It follows from Proposition \eqref{thm:DSTau} that we have combinatorial expansions of the $\tau$-functions as in \eqref{eq:CombExpD},  
which we write here for convenience:
\begin{equation}
\tau_{DS}=\sum_{\vec{\lambda}\in SP}\Pf(\sfa_{\vec{\lambda}})\Pf(\sfd_{\vec{\lambda}}).
\end{equation}
The term  term $\Pf(\sfd_{\vec{\lambda}})$ in this context gives the Cartan coordinate\footnote{See \cite{HB}, Appendix E, or \cite{BHH, HO1, HO2} for the definition of {\em Cartan coordinates} in the context of the BKP hierarchy. In our context, by  ``Cartan coordinate'', we simply mean the Pfaffian expression $\Pf(\sfd_{\vec{\lambda}})$, specifying the initial condition in the isotropic Grassmannian.} of the point in the (isotropic) Grassmannian for $L_S \SO(N)$ corresponding to the given initial condition and to the $N$-tuple of strict partitions $\vec{\lambda}$, while the term $\Pf(\sfa_{\vec{\lambda}})$ specifies the time dependence.

\begin{remark}
Previous expressions \cite{Cafasso2015, Cafasso2017, Cafasso2019} for Drinfeld-Sokolov $\tau$-functions were obtained by expanding the Widom $\tau$-function in minor determinants labeled by all types of partitions and then taking a square root. The expansion in terms of Pfaffians is more intrinsic, since everything is formulated directly on the isotropic Grassmannian defined by the fermions \eqref{ofModes}. Furthermore, it is more efficient, since the set of strict partitions is only a small subset of the full set of partitions.
\end{remark}

\section{Examples of polynomial Drinfeld-Sokolov \hbox{$\tau$-functions}}
\label{drinfeld_sokolov_tau}

The simplest type of Drinfeld-Sokolov $\tau$-functions are polynomials, for which the initial condition matrix is 
$\Psi_-(z)=e^X(z)$ where $X(z)$ is nilpotent, so the expression for $\Psi_-$ truncates 
at finite order and $\sfa,\sfd$ are effectively finite block matrices. We compute here some examples of such polynomial 
$\tau$-functions\footnote{See \cite{KvdL2018, KvdL2019, KvdLRoz2021, VdL2021} for a thorough account of polynomial 
$\tau$-functions of KP, BKP, DKP, mKP and multicomponent KP type.}

\subsection{$B_1^{(1)}$ polynomial $\tau$-function}\label{sec:B1pol}

A first example of a polynomial Drinfeld-Sokolov $\tau$-function  in the case $B_1^{(1)}= L_S \SO(3)$, is obtained
by choosing the following upper or lower triangular initial condition data:
\be
X_1(z)=\frac{1}{z}\left( \begin{array}{ccc}
0 & 0 & 0 \\
a & 0 & 0 \\
0 & a & 0
\end{array} \right), \quad X_2(z)=\frac{1}{z}\left( \begin{array}{ccc}
0 & b & 0 \\
0 & 0 & b \\
0 & 0 & 0
\end{array} \right).
\ee
We restrict to the first three times $t_1,t_3,t_5$, and take $\Psi_+(z)$ in $\eqref{eq:PsipB}$ to be $e^{Y(z, \tb_o)}$, where
\begin{equation}
\begin{split}
Y(z,\tb_o) & :=t_1\Lambda+t_3\Lambda^3+t_5\Lambda^5\\
&=\left(
\begin{array}{ccc}
 0 & \frac{1}{2} z \left(z \left[zt_5 +t_3\right)+t_1\right] & 0 \\
 z \left(zt_5 +t_3\right)+t_1 & 0 & \frac{1}{2} z \left[z \left(zt_5 +t_3\right)+t_1\right] \\
 0 & z \left(zt_5 +t_3\right)+t_1 & 0 ,\\
\end{array}
\right)
\end{split},
\end{equation}
so the factorized group element is
\be
J_a(z) := e^{-X_a(z)}e^{Y(z,\tb_0)}, \quad a =1, 2.
\ee
The Cartan coordinates (the coefficients in the $Q$-Schur function expansion) and time polynomials for the initial condition 
matrix $\Psi_+(z)=e^{X_1(z)}$ are then
\begin{equation}\label{eq:TableT31}
\smallskip
\begin{array}{c|c|c}
\Pf(\sfd_{\vec{\lambda}}) & \Pf(\sfa_{\vec{\lambda}}) & \vec{\lambda} \\
\hline
-a & \frac{t_1}{2}=\frac{1}{2} Q_{(1,0)}(\tb_0) & \left(
\begin{array}{ccc}
 \emptyset , & (0) , & (0)  \\
\end{array}
\right) \\
-\frac{a^2}{2} & -\frac{t_1^2}{8}=-\frac{1}{4} Q_{(2,0)}(\tb_0) & \left(
\begin{array}{ccc}
 \emptyset ,& \emptyset ,& (1,0) \\
\end{array}
\right)
\end{array}
\smallskip
\end{equation}
and the resulting $\tau$-function is
\begin{equation}
\tau_{O}[J_1]=\left(1 - a \frac{t_1}{4}\right)^2.
\end{equation}
The Cartan coordinates and time polynomials for the initial condition matrix $e^{X_2(z)}$ are
\be
\smallskip
\begin{array}{c|c|c}
\Pf(\sfd_{\vec{\lambda}}) & \Pf(\sfa_{\vec{\lambda}}) & \vec{\lambda} \\
\hline
 b & \frac{t_1^3}{12}-t_3=\frac{1}{2}Q_{(2,1)}(\tb_0) & \left(
\begin{array}{ccc}
 (0) , & (0) , & \emptyset \\
\end{array}
\right)\\
-\frac{b^2}{2} & -\frac{1}{288}\left(t_1^3-12 t_3\right)^2=-\frac{1}{2}Q_{(4,2)}(\tb_0) & \left(
\begin{array}{ccc}
 (1,0) ,& \emptyset , & \emptyset \\
\end{array}
\right) 
\end{array},
\smallskip
\ee
and the corresponding $\tau$-function is
\be
\tau_{O}[J_2]=\left(\frac{a}{24} \left(t_1^3-12 t_3\right)+1\right)^2.
\ee

\begin{remark}
Here the $Q_\lambda$'s are $Q$-Schur functions, labelled by {\em strict} partitions \cite{Mac, HB}. In these examples, and others
to  follow, the strict partition $\lambda=\{\lambda_j\}$ labelling the $Q$-Schur function is obtained from the $N$-tuple of strict partitions $\vec{\lambda}=\left(\lambda^{(1)},\dots,\lambda^{(N)} \right)$ through the following (empirical) rule:
\begin{equation}
\{\lambda_i\}=\{N-\alpha +\lambda_j^{(\alpha)}(N-1)\}_{j=1, \dots n_\alpha \atop \alpha=1, \dots, N},
\end{equation}
properly ordered. This is \textit{not} an isomorphism, however, since the same combined strict partition can be obtained from more than one $N$-tuple of strict partitions. 
Moreover, there are, in principle, $N$-tuples of strict partitions that could lead to a combined partition that is non-strict. It seems, however, that those
Pfaffian minors in which this occurs always vanish.
\end{remark}

\begin{remark}\label{SO3SL2}
Comparing with \cite{Cafasso2018}, eq.~(4.1),  the Pfaffian $\tau$-functions for the $SO(3)$ case are themselves 
 seen to be squares of the corresponding $A_1^{(1)}$ Drinfeld-Sokolov polynomial $\tau$-functions with rescaled times. 
 An explanation of this is given in Theorem \ref{thm:SL2toSO3}.
 \end{remark}

\subsection{$B_2^{(1)}$ polynomial $\tau$-function}

We now give a much less trivial example of polynomial $\tau$-function. Specify the initial condition to be $\Psi_-(z)=e^{X_3(z)}$, where
\begin{equation}
X_3(z):=\left(
\begin{array}{ccccc}
 0 & 0 & 0 & 0 & 0 \\
 \frac{a_2}{z} & 0 & 0 & 0 & 0 \\
 \frac{a_3}{z} & \frac{a_5}{z} & 0 & 0 & 0 \\
 \frac{a_4}{z} & 0 & \frac{a_5}{z} & 0 & 0 \\
 0 & \frac{a_4}{z} & -\frac{a_3}{z} & \frac{a_2}{z} & 0 \\
\end{array}
\right),
\end{equation}
and the time evolution factor to be  $\Psi_+(z)=e^{Y_3(z,\tb_o)}$, where
\bea
Y_3(z,\tb_o)&\&:=t_1\Lambda+t_3\Lambda^3+t_5\Lambda^5\cr
&\& \cr
&\&=\left(
\begin{array}{ccccc}
 0 & -\frac{1}{2} zt_3  & 0 & -\frac{1}{2}  \left(z^2t_5+ zt_1\right) & 0 \\
 -(zt_5+t_1) & 0 & -zt_3 & 0 & -\frac{1}{2}  \left(z^2t_5 + z t_1\right) \\
 0 & -(zt_5+t_1) & 0 & -zt_3 & 0 \\
 -t_3 & 0 & -(zt_5+t_1) & 0 & -\frac{1}{2} zt_3 \\
 0 & -t_3 & 0 & -(zt_5+t_1) & 0 \\
\end{array},
\right)\cr
&\&
\eea
so
\be
J_3(z) := e^{-X_3(z)}e^{Y_3(z, \tb_0)}.
\ee
In this case, there are many more nonzero Cartan coordinates. For readability, we tabulate them, and the corresponding time dependant polynomial, 
only up to weights $|\lambda^{\alpha)}|=4$,
in Appendix \ref{B_2_1}

The expressions for coloured strict partitions of higher order become very long, so we do not write them down explicitly. Remarkably, all the time polynomials are again multiples of $Q$-Schur polynomials, which are tabulated, together with their coefficients, in Appendices \ref{B_2_1_a_coeffs} and \ref{B_2_1_d_coeffs} 
 for 5-tuples  $\vec{\lambda}$ of strict partitions of lengths $\ell(\lambda^{\alpha})\le 4$. The $\tau$-function is a polynomial of degree 30,  
 which is too long to write in detail for the general case. However, setting e.g. $a_2=0$ considerably simplifies  the expression, which becomes
\begin{equation}
\begin{split}
\tau_0[J_3] & =1+\frac{a_4 t_1}{2}+\frac{1}{4} a_3 t_1^2-\frac{1}{12} a_5 t_1^3-\frac{1}{192} a_3^2 t_1^4+\frac{1}{192} a_3 a_5 t_1^5-\frac{1}{720} a_5^2 t_1^6\\
&+a_5 t_3-\frac{1}{8} a_3^2 t_3 t_1+\frac{1}{8} a_3 a_5 t_3 t_1^2-\frac{1}{24} a_5^2 t_3 t_1^3+\frac{1}{4} a_5^2 t_3^2+\frac{1}{4} a_5^2 t_5 t_1.
\end{split}
\end{equation}
\subsection{$D_4^{(1)}$ polynomial $\tau$-function}
We conclude our list of examples with a simple polynomial $\tau$-function for the $D_\ell^{(1)}$ series, choosing the case of $SO(8)$. Since the size 
of the matrices starts to be quite large, we consider here the simplest lower triangular initial condition, and as before we consider time evolution with respect to $t_1,t_3,t_5$:\, choosing
\begin{align}
X_4(z)=\frac{1}{z}\left( \begin{array}{cccccccc}
0 & 0 & 0 & 0 & 0 & 0 & 0 & 0 \\
0 & 0 & 0 & 0 & 0 & 0 & 0 & 0 \\
0 & 0 & 0 & 0 & 0 & 0 & 0 & 0 \\
0 & 0 & 0 & 0 & 0 & 0 & 0 & 0 \\
0 & 0 & 0 & 0 & 0 & 0 & 0 & 0 \\
0 & 0 & 0 & 0 & 0 & 0 & 0 & 0 \\
a & 0 & 0 & 0 & 0 & 0 & 0 & 0 \\
0 & a & 0 & 0 & 0 & 0 & 0 & 0 
\end{array} \right), && Y_4(z, \tb_o)= z\left(
\begin{array}{cccccccc}
 0 & \frac{t_5}{2} & 0 & \frac{t_3}{4} & \frac{t_3 }{2} & 0 & \frac{t_1}{2} & 0 \\
 t_1 & 0 & t_5  & 0 & 0 & t_3  & 0 & \frac{t_1 }{2} \\
 0 & t_1 & 0 & \frac{t_5 }{2} & t_5  & 0 & t_3  & 0 \\
 t_3 & 0 & t_1 & 0 & 0 & t_5  & 0 & \frac{t_3 }{2} \\
 \frac{t_3}{2} & 0 & \frac{t_1}{2} & 0 & 0 & \frac{t_5 }{2} & 0 & \frac{t_3 }{4} \\
 0 & t_3 & 0 & \frac{t_1}{2} & t_1 & 0 & t_5  & 0 \\
 t_5 & 0 & t_3 & 0 & 0 & t_1 & 0 & \frac{t_5 }{2} \\
 0 & t_5 & 0 & \frac{t_3}{2} & t_3 & 0 & t_1 & 0 \\
\end{array}
\right).
\end{align}
with
\be
J_4(z) :=e^{-X_4(z)}e^{Y_4(z,\tb_0)}
\ee

The only nontrivial Cartan coordinate is
\begin{equation}
\Pf\left(\sfd_{\left((0),(0),(0),(0),(0),(0),(1),(1) \right)}\right)=-a,
\end{equation}
so we only need
\begin{equation}
\Pf\left(\sfa_{\left((0),(0),(0),(0),(0),(0),(1),(1) \right)}\right)=\frac{t_1}{2}=\frac{1}{2}Q_{(1,0)}(\tb_o),
\end{equation}
giving
\begin{equation}
\tau_O[J_4]=1-a\frac{t_1}{2}.
\end{equation}


\section{$SO(N)$ linear systems and their isomonodromic deformations}
\label{SO_N_isomonodromy}
Another important class of problems for which the $\tau$-function has been identified with a Widom constant \eqref{eq:WidomTauDef} 
are the isomonodromic deformations of $SL(N)$ linear systems of ODEs on the Riemann sphere with punctures \cite{Cafasso2017}.
In the following, we show how the RH problem may be defined, in the case of four simple poles, and how the general solution
for the $\SO(3)$ case may be deduced from the corresponding results derived in \cite{Cafasso2017} for the $SL(2)$ case, corresponding
to Painlevé's sixth transcendant $P_{VI}$.

\subsection{4-point $\mathfrak{g}$-valued linear system and the Widom $\tau$-function}

We illustrate first how to generalize the known construction for $SL(N)$ to the case of an arbitrary matrix representation $\rho$ of a Lie algebra $\mathfrak{g}$ corresponding to a semisimple Lie group $G$, by considering the explicit example of a Fuchsian system with four singular points $0,t,1,\infty$:
\be
\label{eq:4ptsyst}
\partial_z\Phi(z)=\Phi(z)A(z) , \quad A(z)=\frac{A_0}{z}+\frac{A_t}{z-t}+\frac{A_1}{z-1},
\ee
where $A_0,A_t,A_1$ are $N\times N$ matrices in the representation $\rho$. The case of a more general linear systems on the sphere with rational coefficients can be studied by similar means following the construction of \cite{Cafasso2017}, Section 3.

In the generic case, the matrices $A_0,A_t,A_1$ can be written as
\begin{equation}
A_k=G_k\Theta_kG_k^{-1},\quad k=0,t,1,
\end{equation}
where $G_k\in G$ and $\Theta_k\in\mathfrak{h}$, the Cartan subalgebra of $\mathfrak{g}$. The local solution of this linear system around the singular points will be
\begin{equation}
\Phi^{(k)}(z)=
    \begin{cases}
    (z_k-z)^{\Theta_k}G^{(k)}(z), & k=0,1,t, \\
    (-z)^{-\Theta_\infty}G^{(\infty)}(z),
    \end{cases}
\end{equation}
where $G^{(k)}(z)\in LG$ are holomorphic matrix functions in a neighborhood of the singular point, such that $G^{(k)}(z_k)=G_k$. Around the singular points of the equation, $\Phi(z)$ will have monodromies
\begin{align}
M_k=G_k e^{2\pi i\Theta_k}G_k^{-1},\quad k=0,1,t, && M_0M_tM_1M_{\infty}=\mathbb{I}_n.
\end{align}
Now define $\mathfrak{S}$ through
\begin{equation}\label{eq:SigmaDef}
e^{2\pi i\mathfrak{S}}:=M_0M_t,
\end{equation}
which generically can be assumed to lie in the Cartan subalgebra of the Lie algebra (if needed, by applying a constant gauge transformation to the linear system \eqref{eq:4ptsyst}). $\mathfrak{S}$ in \eqref{eq:SigmaDef} is defined only up to shifts in the root lattice $\mathcal{Q}(\mathfrak{g})$. Let $D_a$, $a=0,1,t$ be small discs surrounding the points $0,1,t$, and assume for convenience that $0<|t|<1$. Consider the contour $\Gamma$ shown in Figure \ref{Fig:4ptRHP}. It divides the Riemann sphere with the points $0,t,1,\infty$ removed into two regions, which we call $\T_\pm$, corresponding geometrically to a decomposition of the four-punctured sphere into two pairs of pants which we identify with the regions $\T_\pm$, as in Figure \ref{fig:4ptPants}. The circle $\cC$ of radius $R$, centered at the origin, along which the pants are glued is identified as the red circle of Figure \ref{fig:4point}, with $|t|<R<1$.
\begin{figure}[h!]
\begin{center}
\begin{subfigure}{.4\textwidth}
\centering
\begin{tikzpicture}[scale=1.5]
\draw[thick,decoration={markings, mark=at position 0.27 with {\arrow{>}}},postaction={decorate}](0,0) circle[x radius=0.2, y radius=0.2];
\draw[thick,decoration={markings, mark=at position 0.27 with {\arrow{>}}},postaction={decorate}](0.6,0) circle[x radius=0.2, y radius=0.2];
\draw[thick,decoration={markings, mark=at position 0.25 with {\arrow{<}}},postaction={decorate}](0,0) circle[x radius=1, y radius=1];
\draw[thick,,decoration={markings, mark=at position 0.25 with {\arrow{>}}},postaction={decorate},color=red](0,0) circle[x radius=1.25, y radius=1.25];
\draw[thick,decoration={markings, mark=at position 0.25 with {\arrow{>}}},postaction={decorate}](0,0) circle[x radius=1.5, y radius=1.5];
\draw[thick,decoration={markings, mark=at position 0.27 with {\arrow{>}}},postaction={decorate}](2,0) circle[x radius=0.2, y radius=0.2];
\draw[thick,decoration={markings, mark=at position 0.25 with {\arrow{<}}},postaction={decorate}](0,0) circle[x radius=2.5, y radius=2.5];

\draw[thick,decoration={markings, mark=at position 0.75 with {\arrow{>}}},postaction={decorate}] (0.2,0) to (0.4,0);
\draw[thick,decoration={markings, mark=at position 0.75 with {\arrow{>}}},postaction={decorate}] (0.8,0) to (1,0);
\draw[thick,decoration={markings, mark=at position 0.75 with {\arrow{>}}},postaction={decorate}] (1.5,0) to (1.8,0);
\draw[thick,decoration={markings, mark=at position 0.75 with {\arrow{>}}},postaction={decorate}] (2.2,0) to (2.5,0);

\node at (0,0) {\tiny$\times$};
\node at (0,-0.4) {$D_0$};
\node at (0.6,0) {\tiny$\times$};
\node at (0.6,-0.4) {$D_t$};
\node at (2,0) {\tiny$\times$};
\node at (2,-0.4) {$D_1$};
\node at (1.37,0) {$\red\cC$};
\node at (-0.5,0) {\Large $\mathcal{T}_-$};
\node at (-2,0) {\Large $\mathcal{T}_+$};

\end{tikzpicture}
\caption{The black solid lines constitute the RHP contour $\Gamma$, dividing the complex plane in two regions $\mathcal{T}_\pm$.}\label{Fig:4ptRHP}
\end{subfigure}
\hfill
\begin{subfigure}{.4\textwidth}
\centering
\begin{tikzpicture}[scale=1.3]
\draw(-2,1) circle[x radius=0.2, y radius=0.5];
\draw(-2,-1) circle[x radius=0.2, y radius=0.5];
\draw(-2,0.5) to [out=0,in=90] (-1.5,0) to [out=270,in=0] (-2,-0.5);
\draw(-2,1.5) to[out=0,in=180] (0,0.5);
\draw(-2,-1.5) to[out=0,in=180] (0,-0.5);

\fill[red!30!white] (0,-0.5) to[out=135,in=225] (0,0.5)
to (0.17,0.5) to[out=190,in=170] (0.17,-0.5) --cycle;
\draw[dashed,color=black!60!white](0,-0.5) to[out=135,in=225] (0,0.5)
to (0.17,0.5) to[out=190,in=170] (0.17,-0.5) --cycle;
\fill[red](0,-0.5) to[out=0,in=0] (0,0.5)
to (0.1,0.5) to[out=10,in=-10] (0.1,-0.5 ) --cycle;
\draw(0,-0.5) to[out=0,in=0] (0,0.5)
to (0.1,0.5) to[out=10,in=-10] (0.1,-0.5) --cycle;

\draw(2,1) circle[x radius=0.2, y radius=0.5];
\draw(2,-1) circle[x radius=0.2, y radius=0.5];
\draw(2,0.5) to [out=180,in=90] (1.5,0) to [out=270,in=180] (2,-0.5);
\draw(2,1.5) to[out=180,in=0] (0,0.5);
\draw(2,-1.5) to[out=180,in=0] (0,-0.5);

\node at ($(-2,-1)$) {$0$};
\node at ($(-2,1)$) {$t$};
\node at ($(2,1)$) {$1$};
\node at ($(2,-1)$) {$\infty$};
\node at ($(0.1,0)$) {$\cC$};
\node at (-1,0) {\Large $\mathcal{T}_-$};
\node at (1,0) {\Large $\mathcal{T}_+$};

\end{tikzpicture}
\caption{Pants decomposition of four-punctured sphere corresponding to the contour $\Gamma$. The two pants are glued along the circle $\cC$.}
\label{fig:4ptPants}
\end{subfigure}
\end{center}
\caption{}
\label{fig:4point}
\end{figure}

To a fundamental matrix solution of \eqref{eq:4ptsyst} we can associate a piecewise defined function
\begin{equation}
\Psi(z):= \begin{cases}
G^{(a)}(z), & z\in D_a,\quad, a=0,1,t, \\
\Phi(z), &z\in\mathcal{T}_-\cup\mathcal{T}_+,\, z\notin,\Gamma, \\
(-z)^{\mathfrak{S}}\Phi(z), & z\notin\Gamma.
\end{cases}
\end{equation}
This function solves a \textit{dual} RHP (i.e. with appropriate jumps factorized as in \eqref{eq:JumpDual}) on $\Gamma$, which we do not specify since they will not be needed in the following (see \cite{Gavrylenko2016b, Cafasso2017} for further details).

To apply the Widom constant formalism, we need to reduce this dual Riemann-Hilbert problem to a \textit{direct} Riemann-Hilbert problem on a circle with factorization as in \eqref{eq:JumpRHP}, written in terms of known functions $\Psi_{\pm}$. To do this, consider the solutions $\Phi_\pm$ of two auxiliary 3-point Fuchsian linear systems 
on the two trinions (i.e. {\em pairs of pants}, obtained by replacing the three punctures at the poles by oval boundaries).
\begin{align}
\partial_z\Phi_-(z)=\Phi_-(z)A_-(z)\quad A_-(z)=\frac{A_{0,-}}{z}+\frac{A_{t,-}}{z-t}, \nonumber\\
\partial_z\Phi_+(z)=\Phi_+(z)A_+(z)\quad A_+(z)=\frac{A_{0,+}}{z}+\frac{A_{1,+}}{z-1},
\end{align}
where $A_{k,\pm}\sim\Theta_k$ are $N\times N$ matrices in the representation $\rho$ of $\mathfrak{g}$, normalized as
\begin{align}
\Phi_+(z)\simeq (-z)^{\mathfrak{S}}, \quad z\rightarrow 0, && \Phi_-(z)\simeq (-z)^{\mathfrak{S}},\quad z\rightarrow\infty.
\end{align}
By restricting the dual RHP to the pants $\T_\pm$ in Figure \ref{fig:4point} we obtain two solutions $\Psi_\pm$ of 
two three-point problems, which on the circle $\mathcal{C}$ are related to $\Phi_{\pm}$ by
\begin{equation}
\Psi_\pm(z)=(-z)^{-\mathfrak{S}}\Phi_{\pm}(z),\qquad z\in\cC.
\end{equation}
Since $\Psi_\pm$ have the same jumps as $\Psi$ inside the trinions $\T_{\pm}$ respectively, the function
\begin{equation}
\bar{\Psi}(z):= \begin{cases}
\Psi_{+}(z)^{-1}\Psi(z):=\bar{\Psi}_-(z), & \text{outside } \cC, \\
\Psi_{-}(z)^{-1}\Psi(z):=\bar{\Psi}_+(z), & \text{inside } \cC
\end{cases}
\end{equation}
is single-valued everywhere apart from the circle $\cC$, where it has a jump admitting the two factorizations\footnote{Note that in the previous sections we dealt with Riemann-Hilbert problems on the unit circle, while here the circle has radius $R$. Everything is easily generalized to this case by reinserting the radius $R$ in the expressions where needed. All our quantities depend only on the splitting of the space $L^2(S^1)$ into the subspace $\HH^N_+$ of functions admitting analytic continuation inside the circle and those $\HH^N_-$ admitting analytic continuation outside.}
\begin{equation}
J(z)=\bar{\Psi}_+(z)\bar{\Psi}_-(z)^{-1} =\Psi_-(z)^{-1}\Psi_+(z).
\end{equation}

In \cite{Cafasso2017} it was shown that the Widom $\tau$-function \eqref{eq:WidomTauDef}, with kernel written in terms of the 
functions $\Psi_\pm$ defined above, is related to the JMU $\tau$-function for the linear system \eqref{eq:4ptsyst}, defined by
\begin{equation}\label{eq:JMUSLN}
\partial_t\log\tau_{JMU}^{\SL(N)}=\frac{1}{2}\res_{z=t}\tr A(z)^2=\frac{\tr A_0A_t}{t}+\frac{A_tA_1}{t-1},
\end{equation}
in the following way:
\begin{equation}
\tau_{JMU}^{\SL(N)}(t)=\text{const.}\,t^{\frac{1}{2}\tr\left(\mathfrak{S}^2-\Theta_0^2-\Theta_t^2\right)}\tau_W(t).
\end{equation}
Although the proof in \cite{Cafasso2017}, based on comparing the logarithmic $t$-derivative of $\tau_W$ to \eqref{eq:JMUSLN}, 
was given only for the $\SL(N)$ case, it applies to this more general context without any modification, if the RHP is formulated analogously
to the above. 
For an arbitrary classical Lie group $G$, the definition of the JMU $\tau$-function can be expressed as
\begin{equation}
\partial_t\log\tau_{\JMU}^{G}=\frac{\kappa}{2}\res_{z=t}\tr\left( A(z)^2\right),
\end{equation}
where $\kappa=1$ for $\SL(N)$ and $\Sp(N)$ and $\kappa=\frac{1}{2}$ for $SO(N)$. Specializing to $SO(N)$, we have
\begin{equation}
\begin{split}
\partial_t\log\tau_{\JMU}^{SO(N)} & =\frac{1}{4}\res_{z=t}\tr A(z)^2=\frac{\tr A_0A_t}{2t}+\frac{A_tA_1}{2(t-1)}=\frac{1}{2} \partial_t\log\left(t^{\frac{1}{2}\tr\left(\mathfrak{S}^2-\Theta_0^2-\Theta_t^2\right)}\tau_W(t) \right) \\
& = \partial_t\log\left(t^{\frac{1}{4}\tr\left(\mathfrak{S}^2-\Theta_0^2-\Theta_t^2\right)}\tau_O[J](t) \right),
\end{split}
\end{equation}
and
\begin{equation}\label{eq:JMUO}
\tau_{\JMU}^{SO(N)}=\text{const.}\,t^{\frac{1}{4}\tr\left(\mathfrak{S}^2-\Theta_0^2-\Theta_t^2\right)}\tau_O[J].
\end{equation}
This may be generalized to an arbitrary number of Fuchsian singularities, 
or to irregular singularities of higher Poincar\'e rank at $0$ and $\infty$,  
following the $\SL(N)$ construction of \cite{Cafasso2017}.

\subsection{The $SO(3)$ isomonodromic $\tau$-function}
In the case of an $\SL(2,\mathbb{C})$ linear system, it is well-known that the solution to the 3-point problem can be written explicitly in terms of hypergeometric functions $_2F_1$. We recall the expressions for the solution of the corresponding $\SL(2)$ RHP, which we denote by $\Psi_\pm^{(2)}$, from \cite{Gavrylenko2016b}: {\footnotesize
\begin{equation}\label{eq:SL2solp}
\begin{split}
\Psi^{(2)}_+(z) & =\left( \begin{array}{cc}
\pFq{1}{2}{\theta_1+\theta_\infty+\sigma,\theta_1-\theta_\infty+\sigma}{2\sigma}{z} & z\frac{\theta_\infty^2-(\theta_1+\sigma)^2}{2\sigma(1+2\sigma)}\pFq{1}{2}{1+\theta_1+\theta_\infty+\sigma,1+\theta_1-\theta_\infty+\sigma}{2+2\sigma}{z} \\
-z\frac{\theta_\infty^2-(\theta_1-\sigma)^2}{2\sigma(1-2\sigma)}\pFq{1}{2}{1+\theta_1+\theta_\infty-\sigma,1+\theta_1-\theta_\infty-\sigma}{2-2\sigma}{z} & \pFq{1}{2}{\theta_1+\theta_\infty-\sigma,\theta_1-\theta_\infty-\sigma}{-2\sigma}{z}
\end{array} \right)\\
& := \left( \begin{array}{cc}
K_{++}(z) & K_{+-}(z) \\
K_{-+}(z) & K_{--}(z)
\end{array} \right) .
\end{split}
\end{equation}
\begin{equation}\label{eq:SL2solm}
\begin{split}
\Psi^{(2)}_-(z) & =\left( \begin{array}{cc}
\pFq{1}{2}{\theta_t+\theta_0-\sigma,\theta_t-\theta_0-\sigma}{-2\sigma}{\frac{t}{z}} & -t^{-2\sigma}e^{-i\eta}\frac{t}{z}\frac{\theta_0^2-(\theta_t-\sigma)^2}{2\sigma(1-2\sigma)}\pFq{1}{2}{1+\theta_t+\theta_0-\sigma,1+\theta_t-\theta_0-\sigma}{2-2\sigma}{\frac{t}{z}} \\
t^{2\sigma}e^{i\eta}\frac{t}{z}\frac{\theta_0^2-(\theta_t+\sigma)^2}{2\sigma(1+2\sigma)}\pFq{1}{2}{1+\theta_t+\theta_0+\sigma,1+\theta_t-\theta_0+\sigma}{2+2\sigma}{\frac{t}{z}} & \pFq{1}{2}{\theta_t+\theta_0+\sigma,\theta_t-\theta_0+\sigma}{2\sigma}{\frac{t}{z}}
\end{array} \right) \\
&:= \left( \begin{array}{cc}
\widetilde{K}_{++}(z) & \widetilde{K}_{+-}(z) \\
\widetilde{K}_{-+}(z) & \widetilde{K}_{--}(z)
\end{array} \right) .
\end{split} .
\end{equation} }

While the $SO(3)$ $3$-point linear system might at first seem  harder to solve, we can still express its solution 
in terms of hypergeometric functions, simply by changing the representation.
\begin{proposition}\label{prop:SL2toSO3}
Let $\Phi^{(2)}(z)$ be the solution to the $2\times 2$ linear system
\be
\partial_z\Phi^{(2)}(z)=\Phi^{(2)}(z)A^{(2)}(z), \quad A(z)=\sum_{\alpha=1}^3A^{\alpha}(z)\sigma_\alpha,
\ee
where $A^a(z)$ are rational functions and $\{\sigma_\alpha\}_{\alpha=1,2,3}$ are the Pauli matrices. 
The $3\times 3$ matrix $\Phi(z)$ with components
\begin{equation}\label{eq:SL2toSO3}
\Phi_{\alpha\beta}(z)=\frac{1}{2}\tr\left(\sigma_i\Phi^{(2)}(z)\sigma_j\Phi^{(2)}(z)^{-1} \right),\quad\alpha,\beta=1,\dots,3
\end{equation}
satisfies $\Phi(z)\Phi(z)^t=\mathbb{I}_3$ and solves the $3\times 3$ linear system with rational coefficients
\be
\label{eq:SL2toSO3sys}
\partial_z\Phi(z)=\Phi(z)A^{(3)}(z), \quad A^{(3)}(z)=-2i\sum_{\alpha=1}^3A^{\alpha}(z)J_\alpha ,
\ee
where we have introduced the $SO(3)$ generators
\begin{equation}
(J_\alpha)_{\beta\gamma}=\epsilon_{\alpha\beta\gamma}.
\end{equation}
\end{proposition}
\begin{proof}
This follows from the isomorphism between representations of $\SL(2)$ and $SO(3)$ which t can be verified directly 
using the Euler angle representation
\begin{equation}
\Phi^{(2)}(z)=e^{\varphi_1(z)\sigma_3}e^{\varphi_2(z)\sigma_2}e^{\varphi_3(z)\sigma_3},
\end{equation}
which allows us to verify \eqref{eq:SL2toSO3sys} by explicit computation.
\end{proof}

Proposition \ref{prop:SL2toSO3} gives  a way to compute the $\tau$-function for the $SO(3)$ $4$-point isomonodromic problem; just transform the $\SL(2)$ solutions $\Psi_\pm^{(2)}$ into $SO(3)$ solutions $\Psi_\pm^{(3)}$ using \eqref{prop:SL2toSO3}, which then can be used to define the Pfaffian $\tau$-function \eqref{eq:TauFredPf} through the operators $\sfa,\sfd$ defined in eqs.~\eqref{eq:adKers}. The resulting $3\times 3$ kernels will have explicit expressions in terms of bilinear combinations of hypergeometric functions. Proposition \ref{prop:SL2toSO3} allows us to write an explicit Pfaffian $\tau$-function whenever the corresponding $\SL(2)$ kernel entering in the Widom $\tau$-function is known. This includes the case of Painlev\'e $VI$ or $III$ equations, as well as the $\SL(2)$ Schlesinger system
 \cite{Gavrylenko2016b, Gavrylenko2017, Cafasso2017}.

In fact,  there is an even more direct relation between the $SO(3)$ $\tau$-function and the $\SL(2)$ $\tau$-function:
\begin{theorem}\label{thm:SL2toSO3}
Let $\tau_W[J^{(2)}]$ be the Widom $\tau$-function \eqref{eq:WidomTauDef}, associated to a Riemann-Hilbert problem 
on a circle with jump $J^{(2)}(z)$, admitting direct and dual factorizations
\begin{equation}\label{eq:SL2Jump}
J^{(2)}(z)=\Psi_-^{(2)}(z)^{-1}\Psi_+^{(2)}(z)=\bar{\Psi}^{(2)}_+(z)\bar{\Psi}^{(2)}_-(z)^{-1}\in L\SL(2),
\end{equation}
and let $\tau_O[J^{(3)}]$ be the $SO(3)$ $\tau$-function \eqref{eq:TauD}, associated to a RHP whose solution is 
obtained from \eqref{eq:SL2Jump} through the map \eqref{eq:SL2toSO3}. Then
\begin{equation}
\tau_O[J^{(3)}]=\left(\tau_W[J^{(2]}]\right)^2.
\end{equation}
\end{theorem}
\begin{proof}
If $\xi$ is any parameter (which can be an isospectral or isomonodromic time, or a monodromy coordinate) that $J^{(2)}$ depends upon, the Widom $\tau$-function satisfies (see Theorem 2.3 in \cite{Cafasso2017}):
\begin{equation}
\partial_\xi\log\tau_W[J^{(2)}]=\frac{1}{2\pi i}\oint_{S^1}\tr\left\{J^{(2)}(z)^{-1}\partial_\xi J^{(2)}(z)\left[\partial_z\bar{\Psi}^{(2)}_-(z)\bar{\Psi}_-^{(2)}(z)^{-1}+\Psi_+^{(2)}(z)^{-1}\partial_z\Psi_+^{(2)}(z) \right] \right\} dz.
\end{equation}
On the other hand, the $\tau$-function $\tau_O[J_3]$ will satisfy
\begin{equation}
\partial_\xi\log\tau_O[J^{(3)}]=\frac{1}{4\pi i}\oint_{S^1}\tr\left\{J^{(3)}(z)^{-1}\partial_\xi J^{(3)}(z)\left[\partial_z\bar{\Psi}^{(3)}_-(z)\bar{\Psi}_-^{(3)}(z)^{-1}+\Psi_+^{(3)}(z)^{-1}\partial_z\Psi_+^{(3)}(z) \right] \right\} dz.
\end{equation}
The theorem follows from the relation
\begin{equation}
\begin{split}
&\tr\left\{J^{(3)}(z)^{-1}\partial_t J^{(3)}(z)\left[\partial_z\bar{\Psi}^{(3)}_-(z)\bar{\Psi}_-^{(3)}(z)^{-1}+\Psi_+^{(3)}(z)^{-1}\partial_z\Psi_+^{(3)}(z) \right] \right\}\\
& =4\tr\left\{J^{(2)}(z)^{-1}\partial_t J^{(2)}(z)\left[\partial_z\bar{\Psi}^{(2)}_-(z)\bar{\Psi}_-^{(2)}(z)^{-1}+\Psi_+^{(2)}(z)^{-1}\partial_z\Psi_+^{(2)}(z) \right] \right\},
\end{split}
\end{equation}
which implies that for any arbitrary parameter $\xi$ on which the $\tau$-function depends, we have
\begin{equation}
\partial_\xi\log\tau_O[J^{(3)}]=2\partial_\xi\log\tau_W[J^{(2)}].
\end{equation}
\end{proof}
As an application of this theorem, using equations \eqref{eq:JMUSLN}, \eqref{eq:JMUO}, we obtain the isomonodromic $\tau$-function for the 4-point Fuchsian system on the sphere as the square of the Painlev\'e VI $\tau$-function. This also explains the remark made at the end of Section \ref{sec:B1pol}, that 
the Drinfeld-Sokolov polynomial $\tau$-functions of $B_1^{(1)}$ are squares of the corresponding $\tau$-functions of $A_1^{(1)}$.

\section{Summary and future directions}

We have shown that Widom $\tau$-functions derived from Riemann-Hilbert problems in orthogonal loop groups have Pfaffian rather than determinantal expressions, and provided an explicit formula \eqref{eq:CombExpD} for their combinatorial expansion. Using the same methods one could study other types of Drinfeld-Sokolov $\tau$-functions, such as topological $\tau$-functions of Kontsevich-Witten type in the orthogonal case \cite{fan2016witten, liu2015bcfg}. It remains to find a deeper explanation for the appearance of $Q$-Schur functions in the expansions, which suggests a relation to the BKP hierarchy \cite{jimbo_miwa_83,HB} through a suitable reduction procedure.

Computing the combinatorial expansion of $SO(N)$ isomonodromic $\tau$-function should yield an orthogonal variant of the ``Kiev formula''
\cite{Gavrylenko2016b, Gavrylenko2018a,del2020isomonodromic}, expressing isomonodromic $\tau$-functions for $\SL(N)$ linear systems in terms of conformal blocks of a $W(\mathfrak{sl}_N)$ algebra. The natural analogue would be conformal blocks of $W(\mathfrak{so}_N)$ algebras, for which no expression is available in the case of generic monodromies (see, however, \cite{Bershtein:2017lmg} for $\mathfrak{so}_N$ quasi-permutation monodromies, and \cite{Bonelli2021b} for nonautonomous Toda chains).

Another interesting variant would be the case of Riemann-Hilbert problems taking values in \textit{twisted} loop groups, using multicomponent {\em neutral}
fermions.  Problems of this type appear in the study of Frobenius manifolds and topological quantum field theories \cite{Dubrovin1994, feigin2010givental}.

\break

\appendix


\section{Appendix. A polynomial $\tau$-function of $B^{(1)}_2$ type}
\label{B_2_1}
Recall equation \eqref{eq:stricthi} in the main text, representing the map between decreasing sequences of positive half-integers $\{p_{\alpha,i}\}_{i=1}^{n_\alpha}$ and strict partitions
\begin{equation}
\lambda^{(\alpha)}=\left(p_{\alpha,1}-\frac{1}{2},\dots,p_{\alpha,n_{\alpha}}-\frac{1}{2}\right) .
\end{equation}
In the tables below we  group the coefficients in the Pfaffian minor expansions by the
total {\em weight} of the $N$-tuple of (positive) strict partitions obtained by adding $1$ to each part:  
\begin{equation}
|\vec{\lambda}|:=\sum_{\alpha=1}^N\sum_{i=1}^{n_\alpha}(p_{\alpha,i}+\frac{1}{2}) ,
\end{equation}
counting the overall power of $z$ (resp. $z^{-1}$) in the Fourier expansion of $\sfa$ (respectively $\sfd$) that contribute to the Pfaffian  minor. For example
\begin{equation}
\Pf\left( \begin{array}{cc}
0 & \left(S\sfa^{\frac{1}{2}}_{-\frac{1}{2}} \right)_{12} \\
\left(S\sfa^{\frac{1}{2}}_{-\frac{1}{2}} \right)_{12} & 0
\end{array} \right)
\end{equation}
has weight $|\vec{\lambda}|=2$, since it contains two Fourier coefficients of $\sfa$ of weight  $0$ that multiply $z^0w^0$, while
\begin{equation}
\Pf\left( \begin{array}{cc}
0 & \left(S\sfa^{\frac{3}{2}}_{-\frac{3}{2}} \right)_{12} \\
\left(S\sfa^{\frac{3}{2}}_{-\frac{3}{2}} \right)_{12} & 0
\end{array} \right)
\end{equation}
has weight $|\vec{\lambda}|=4$, since it contains two Fourier coefficients of $\sfa$ of weight $2$, that multiply $zw$.


\subsection{The Pfaffian $\sfd$ coefficients}
\label{B_2_1_d_coeffs}

\bigskip
{\small \begin{equation}
\begin{array}{c|c}
\Pf(\sfd_{\vec{\lambda}}) & \vec{\lambda},\,|\vec{\lambda}|=2 \\
\hline
 -a_2\, & \left(
\begin{array}{ccccc}
 \emptyset , & (0) , & \emptyset , & \emptyset , & (0)  \\
\end{array}
\right) 
\\
 a_5\, & \left(
\begin{array}{ccccc}
 \emptyset , & \emptyset , & (0) , & (0) , & \emptyset  \\
\end{array}
\right)
 \\
 -a_3\, & \left(
\begin{array}{ccccc}
 \emptyset , & \emptyset , & (0) , & \emptyset , & (0)  \\
\end{array}
\right) 
\\
 -a_4\, & \left(
\begin{array}{ccccc}
 \emptyset , & \emptyset , & \emptyset , & (0) , & (0)  \\
\end{array}
\right) 
\\
\hline
\phantom{\Pf(\sfd_{\vec{\lambda}})} & \vec{\lambda},\,|\vec{\lambda}|=3 \\
\hline
 \frac{a_5^2}{2}\, &\, \left(
\begin{array}{ccccc}
 \emptyset , & \emptyset , & \emptyset , & (1,0) , & \emptyset  \\
\end{array}
\right) 
\\
 \frac{1}{2} \left(a_3^2-2 a_2 a_4\right)\, &\, \left(
\begin{array}{ccccc}
 \emptyset , & \emptyset , & \emptyset , & \emptyset , & (1,0)  \\
\end{array}
\right)
 \\
 -\frac{1}{2} a_2 a_5\, &\, \left(
\begin{array}{ccccc}
 \emptyset , & \emptyset , & (1) , & \emptyset , & (0)  \\
\end{array}
\right) 
\\
 \frac{a_2 a_5}{2}\, &\, \left(
\begin{array}{ccccc}
 \emptyset , & \emptyset , & (0) , & \emptyset , & (1)  \\
\end{array}
\right) 
\\
 -\frac{1}{2} a_3 a_5\, &\, \left(
\begin{array}{ccccc}
 \emptyset , & \emptyset , & \emptyset , & (1) , & (0)  \\
\end{array}
\right) \\
 \frac{a_3 a_5}{2}\, &\, \left(
\begin{array}{ccccc}
 \emptyset , & \emptyset , & \emptyset , & (0) , & (1)  \\
\end{array}\right) \\
\hline
\phantom{\Pf(\sfd_{\vec{\lambda}})} & \vec{\lambda},\,|\vec{\lambda}|=4 \\
\hline
 -\frac{1}{6} a_2 a_5^2 & \left(
\begin{array}{ccccc}
 \emptyset , & \emptyset , & \emptyset , & (2) , & (0)  \\
\end{array}
\right) \\
 -\frac{1}{6} a_2 a_5^2 & \left(
\begin{array}{ccccc}
 \emptyset , & \emptyset , & \emptyset , & (0) , & (2)  \\
\end{array}
\right) \\
\frac{1}{3} a_2 a_5^2 & \left(
\begin{array}{ccccc}
 \emptyset , & \emptyset , & \emptyset , & (1) , & (1)  \\
\end{array}
\right) \\
 -a_2 a_5 & \left(
\begin{array}{ccccc}
 \emptyset , & (0) , & (0) , & (0) , & (0)  \\
\end{array}
\right) 
\nonumber
\\
\end{array}
\end{equation} 
}

\subsection{The Pfaffian $\sfa$ coefficients}
\label{B_2_1_a_coeffs}

\bigskip

{\small
\begin{align}
\begin{array}{c|c}
\Pf(\sfa_{\vec{\lambda}}) & \vec{\lambda},\,|\vec{\lambda}|=2 \\
\hline
 -\frac{1}{12} \left(t_1^3+6 t_3\right)=-\frac{1}{2}Q_{(3,0)} & \left(
\begin{array}{ccccc}
 \emptyset , & (0) , & \emptyset , & \emptyset , & (0)  \\
\end{array}
\right) 
\\
 t_3-\frac{t_1^3}{12}=-\frac{1}{2}Q_{(2,1)} &  \left(
\begin{array}{ccccc}
 \emptyset , & \emptyset , & (0) , & (0) , & \emptyset  \\
\end{array}
\right) 
\\
 -\frac{t_1^2}{4}=-\frac{1}{2}Q_{(2,0)} & \left(
\begin{array}{ccccc}
 \emptyset , & \emptyset , & (0) , & \emptyset , & (0)  \\
\end{array}
\right) 
\\
 -\frac{t_1}{2}=-\frac{1}{2}Q_{(1,0)} & \left(
\begin{array}{ccccc}
 \emptyset , & \emptyset , & \emptyset , & (0) , & (0)  \\
\end{array}
\right)
 \\
\hline
 & \vec{\lambda},\,|\vec{\lambda}|=3 \\
\hline
\frac{1}{360} \left(-t_1^6-30 t_3 t_1^3+180 t_5 t_1+180 t_3^2\right)=-\frac{1}{2}Q_{(5,1)} & \left(
\begin{array}{ccccc}
 \emptyset , & \emptyset , & \emptyset , & (1,0) , & \emptyset  \\
\end{array}
\right)
 \\
 -\frac{1}{96} t_1 \left(t_1^3+24 t_3\right)=-\frac{1}{4}Q_{(4,0)} & \left(
\begin{array}{ccccc}
 \emptyset , & \emptyset , & \emptyset , & \emptyset , & (1,0)  \\
\end{array}
\right) 
\\
 -\frac{t_1^6}{1440}-\frac{1}{12} t_3 t_1^3-\frac{t_5 t_1}{2}-\frac{t_3^2}{4}=-\frac{1}{2}Q_{(6,0)} & \left(
\begin{array}{ccccc}
 \emptyset , & \emptyset , & (1) , & \emptyset , & (0)  \\
\end{array}
\right) 
\\
 \frac{1}{576} \left(t_1^3-12 t_3\right)^2=\frac{1}{4}Q_{(4,2)} & \left(
\begin{array}{ccccc}
 \emptyset , & \emptyset , & (0) , & \emptyset , & (1)  \\
\end{array}
\right) 
\\
 -\frac{t_1^5}{240}-\frac{1}{4} t_3 t_1^2-\frac{t_5}{2}=-\frac{1}{2}Q_{(5,0)} & \left(
\begin{array}{ccccc}
 \emptyset , & \emptyset , & \emptyset , & (1) , & (0)  \\
\end{array}
\right) 
\\
 \frac{1}{160} \left(t_1^5-80 t_5\right)=\frac{1}{4}Q_{(4,1)} & \left(
\begin{array}{ccccc}
 \emptyset , & \emptyset , & \emptyset , & (0) , & (1)  \\
\end{array}
\right)
 \\
\hline
 & \vec{\lambda},\,|\vec{\lambda}|=4 \\
\hline
 -\frac{t_1^9}{725760}-\frac{t_3 t_1^6}{1440}-\frac{1}{48} t_5 t_1^4-\frac{1}{24} t_3^2 t_1^3-\frac{1}{2} t_3 t_5 t_1-\frac{t_3^3}{12}=-\frac{1}{2}Q_{(9,0)} & \left(
\begin{array}{ccccc}
 \emptyset , & \emptyset , & \emptyset , & (2) , & (0)  \\
\end{array}
\right) 
\\
 \frac{t_1^9}{207360}+\frac{1}{720} t_3 t_1^6+\frac{1}{48} t_5 t_1^4+\frac{1}{48} t_3^2 t_1^3-\frac{1}{4} t_3 t_5 t_1-\frac{t_3^3}{12}=\frac{1}{4}Q_{(8,1)} &  \left(
\begin{array}{ccccc}
 \emptyset , & \emptyset , & \emptyset , & (0) , & (2)  \\
\end{array}
\right) 
\\
 -\frac{t_1^9}{103680}+\frac{t_3 t_1^6}{2880}+\frac{1}{96} t_5 t_1^4-\frac{1}{24} t_3^2 t_1^3+\frac{1}{4} t_3 t_5 t_1+\frac{t_3^3}{6}=-\frac{1}{4}Q_{(5,4)} &  \left(
\begin{array}{ccccc}
 \emptyset , & \emptyset , & \emptyset , & (1) , & (1)  \\
\end{array}
\right)
 \\
 \frac{t_1^6-60 t_3 t_1^3+720 t_5 t_1-720 t_3^2}{1440}=\frac{1}{4}Q_{(3,2,1,0)} & \left(
\begin{array}{ccccc}
 \emptyset , & (0) , & (0) , & (0) , & (0)  \\
\end{array}
\right) 
\nonumber \\
\end{array}
\end{align}
}

\newpage
\section{Appendix.  Matrix representation of orthogonal affine Lie algebras}
\label{sec:MatRep}

In the following, $e_{ij}$ is the elementary matrix $(e_{ij})_{\alpha\beta}:=\delta_{\alpha i}\delta_{\beta j}$.

\subsection{Matrix realization of $B_\ell^{(1)}$}
$\bullet$ Weyl generators:
\bea
E_0&\&=\frac{1}{2}(e_{1,2\ell}+e_{2,2\ell+1}), \quad F_0=2\left(e_{2\ell,1}+e_{2\ell+1,2} \right), \cr
 H_0 &\& =e_{1,1}+e_{2,2}-e_{2l,2l}-e_{2l+1,2l+1}, \cr
E_i&\&=e_{i+1,i}+e_{2l+2-i,2l+1-i}, \quad F_i=e_{i,i+1}+e_{2\ell+1-i,2\ell+2-i} \cr
H_i&\&=-e_{i,i}+e_{i+1,i+1}-e_{2\ell+1-i,2\ell+1-i}+e_{2\ell+2-i,2\ell+2-i}, \quad  i=1,\dots,\ell-1, \cr
E_\ell&\&=e_{\ell+1,\ell}+e_{\ell+2,\ell+1}, \quad F_{\ell}=e_{\ell,\ell+1}+e_{\ell+1,\ell+2}, \quad  H_\ell=-e_{\ell,\ell}+e_{\ell+2,\ell+2}.
\eea
$\bullet$ Cartan matrix:
\begin{equation}
A=\left( \begin{array}{cccccc}
2 & 0 & -1 & 0 & \cdots & 0 \\
0 & 2 & -1 & & & \vdots  \\
-1 & -1 & 2 & \ddots & & 0 \\
0 & &  \ddots & \ddots & -1 & 0 \\
\vdots & & & -1 & 2 & -1 \\
0 & \cdots & 0 & 0 & -2 & 2
\end{array} \right)_{(\ell+1)\times(\ell+1)},
\end{equation}
$\bullet$ Chevalley involution:
\begin{equation}\label{eq:ChevalleyB}
S=\mathrm{antidiag}\left(1,-1,1,\dots,-1,1 \right)_{2\ell+1\times2\ell+1}=\left( \begin{array}{cccccc}
0 & 0 & \cdots &  &    & 1 \\
0 & 0 & \cdots &  & -1 & 0 \\
0 & 0 & \cdots & 1 & 0 & 0 \\
0 & 0 & \ddots & 0 & 0 & 0 \\
0 & -1 & 0 & & \cdots & 0 \\
1 & 0 & 0 & & \cdots &  0
\end{array} \right)_{2\ell+1\times2\ell+1}.
\\
\end{equation}

\newpage


\subsection{Matrix realization of $D_\ell^{(1)}$}
$\bullet$ Weyl generators:
\bea
E_0&\&=\frac{1}{2}(e_{1,2\ell-1}+e_{2,2\ell}), \quad F_0=2\left(e_{2\ell-1,1}+e_{2\ell,2}\right), \cr
 H_0&\& =e_{1,1}+e_{2,2}-e_{2l-1,2l-1}-e_{2l,2l}, \cr
E_i&\&=e_{i+1,i}+e_{2l+1-i,2l-i}, \quad F_i=e_{i,i+1}+e_{2\ell-i,2\ell+1-i} \cr
H_i&\&=-e_{i,i}+e_{i+1,i+1}-e_{2\ell-i,2\ell-i}+e_{2\ell+1-i,2\ell+1-i}, \quad i=1,\dots,\ell-1, \cr
E_\ell&\& =\frac{1}{2}\left(e_{\ell+1,\ell-1}+e_{\ell+2,\ell}\right), \quad F_{\ell}=2\left(e_{\ell-1,\ell+1}+e_{\ell,\ell+2}\right), \cr
H_\ell &\&=-e_{\ell-1,\ell-1}-e_{\ell,\ell}+e_{\ell+1,\ell+1} +e_{\ell+2,\ell+2}.
\eea
$\bullet$ Cartan matrix:
\begin{equation}
A=\left( \begin{array}{cccccc}
2 & 0 & -1 & 0 & \cdots & 0 \\
0 & 2 & -1 & & & \vdots  \\
-1 & -1 & 2 & \ddots & & 0 \\
0 & &  \ddots & \ddots & -1 & -1 \\
\vdots & & & -1 & 2 & 0 \\
0 & \cdots & 0 & -1 & 0 & 2
\end{array} \right)_{(\ell+1)\times(\ell+1)}.
\end{equation}
$\bullet$ Chevalley involution:
\begin{equation}\label{eq:ChevalleyD}
S=\mathrm{antidiag}\left(1,-1,1,\dots,(-1)^{\ell},(-1)^{\ell},(-1)^{\ell+1},\dots,-1,1\right)_{2\ell\times2\ell}.
\end{equation}

\break
 \bigskip
\noindent 
\small{ {\it Acknowledgements.} The authors would like to thank D. Yang, M. Caffasso, O. Lisovyy, P. Gavrylenko
for helpful discussions that contributed much to clarifying the results presented here.
This work was partially supported by the Natural Sciences and Engineering Research Council of Canada (NSERC). 

 \bigskip
\noindent 
\small{ {\it Data sharing.} 
Data sharing is not applicable to this article since no new data were created or analyzed in this study.
\bigskip

\bibliographystyle{JHEPAlph}
\bibliography{Biblio.bib}
\end{document}